\newcommand*{\LONG}{}%
\pdfoutput=1

\ifdefined\LONG
\documentclass[11pt]{article}
\usepackage{amsfonts,amsmath,amssymb,latexsym,comment,amsthm}

\usepackage{fullpage}
\usepackage{mathrsfs}
\usepackage{times}
\usepackage{xspace}
\usepackage{cases}
\usepackage{amsmath} 
\usepackage[colorlinks,linkcolor=blue,filecolor=blue,citecolor=blue,urlcolor=blue,pdfstartview=FitH]{hyperref}
\usepackage[ruled,boxed,commentsnumbered,noresetcount]{algorithm2e}

\def\squarebox#1{\hbox to #1{\hfill\vbox to #1{\vfill}}}

\newtheorem{theorem}{Theorem}

\newtheorem{lemma}{Lemma}

\newtheorem{definition}{Definition}

\newtheorem{claim}{Claim}

\newcommand{\namedref}[2]{\hyperref[#2]{#1~\ref*{#2}}}
\newcommand{\sectionref}[1]{\namedref{Section}{#1}}

\newcommand{\theoremref}[1]{\namedref{Theorem}{#1}}

\newcommand{\figureref}[1]{\namedref{Figure}{#1}}

\newcommand{\lemmaref}[1]{\namedref{Lemma}{#1}}

\newcommand{\algref}[1]{\namedref{Algorithm}{#1}}

\newenvironment{RETHM}[2]{\trivlist \item[\hskip\labelsep{\bf
#1\hskip 5pt\relax\ref{#2}.}]\it}{\endtrivlist}
\newcommand{\rethm}[1]{\begin{RETHM}{Theorem}{#1}}
\newcommand{\repro}[1]{\begin{RETHM}{Proposition}{#1}}
\newcommand{\relem}[1]{\begin{RETHM}{Lemma}{#1}}
\newcommand{\recor}[1]{\begin{RETHM}{Corollary}{#1}}
\newcommand{\reclm}[1]{\begin{RETHM}{Claim}{#1}}

\newcommand{\erethm}{\end{RETHM}}
\newcommand{\erepro}{\end{RETHM}}
\newcommand{\erelem}{\end{RETHM}}
\newcommand{\erecor}{\end{RETHM}}
\newcommand{\ereclm}{\end{RETHM}}

\newcommand{\true}{\mathit{true}}

\newcommand{\ie}{\emph{i.e.,\ }}

\def\beginsmall#1{\vspace{-\parskip}\begin{#1}\itemsep-\parskip}

\newcommand{\dd}[1]{\textbf{\color{blue}
[[[dd: #1]]]}}
\newcommand{\ms}[1]{\textbf{\color{red}
[[[meir: #1]]]}}

\newcommand{\tb}{\makebox[0.6cm]{}}
\newcommand{\due}{\makebox[1cm]{}}

\newcommand{\hide}[1]{}

\newcommand{\commentout}[1]{}

\newcommand{\M}{\mathcal{M}}

\newcommand{\RR}{Reliable Broadcast\xspace}

\usepackage{graphicx}

\graphicspath{ {images/} }

\begin{document}


\title{Possibility and Impossibility of Reliable Broadcast in the Bounded Model}
\author{
Danny Dolev\footnote{Danny Dolev is Incumbent of the Berthold Badler Chair in Computer Science. 
The Rachel and Selim Benin School of Computer Science and Engineering
Edmond J. Safra Campus.
This work was supported in part by the HUJI Cyber Security Research Center in conjunction with the Israel National Cyber Bureau in the Prime Minister's Office. This work was supported in part by The Israeli Centers of Research Excellence (I-CORE) program, (Center  No. 4/11).  
}\\
Hebrew University of Jerusalem\\
Jerusalem, Israel\\
\texttt{dolev@cs.huji.ac.il}
\and
Meir Spielrien\\
Hebrew University of Jerusalem\\
Jerusalem, Israel\\
\texttt{meir.spielrein@mail.huji.ac.il}
}

\maketitle
\thispagestyle{empty}

\setcounter{page}{1}


\begin{abstract}

The Reliable Broadcast concept allows an honest party to send a message to all other parties and to make sure that all honest parties receive this message. In addition, it allows an honest party that received a message to know that all other honest parties would also receive the same message. This technique is important to ensure distributed consistency when facing failures. 

In the current paper, we study the ability to use \RR to consistently transmit a sequence of input values in an asynchronous environment with a designated sender. The task can be easily achieved using counters, but cannot be achieved with a bounded memory facing  failures. We weaken the problem and ask whether the receivers can at least share a common suffix.  We prove that in a standard (lossless) asynchronous system no bounded memory protocol can  guarantee a common suffix at all receivers for every input sequence if a single party might crash. 

We further study the problem facing transient faults and prove that when limiting the problem to transmitting a stream of a single value being sent repeatedly we show a bounded memory self-stabilizing protocol that can ensure a common suffix even in the presence of transient faults and an arbitrary number of crash faults.  We  further prove that this last problem is not solvable in the presence of a single Byzantine fault.
Thus, this problem {\bf separates} Byzantine behavior from crash faults in an asynchronous environment. 

\end{abstract}

\newpage

\section{Introduction}
Many distributed algorithms make use of the  `Reliable Broadcast' technique that allows one party to send a message to all other parties  guaranteeing that if any honest party receives a message, all honest parties will receive the same message.
Reliable Broadcast is somewhat weaker than Consensus since it  allows for more flexibility in dealing with a faulty sender - if the sender is faulty there is no requirement that any message will ever be received by honest parties. 
Therefore, \RR is solvable in an asynchronous environment despite of crash faults.

The current paper studies more deeply the ability to achieve
Reliable Broadcast in a `Bounded Model', in which  the memory of honest parties is bounded by some constant, and the system is asynchronous.
For a single message transmission, Reliable Broadcast is achievable with bounded memory in such a system.
For transmitting a sequence of inbound messages via Reliable Broadcast we prove the impossibility of 
even guaranteeing the delivery of only a common suffix of the sequence of inputs, for some sequences of inputs, given that a single party might crash.

We study two variants of the Bounded Model. In the first variant, in addition to having a bounded memory, the capacity of each link is also bounded, i.e., there is a constant $\bar c$ such that the number of messages simultaneously present over each link never exceeds $\bar c$.
In the second variant, in addition to having a bounded memory, the links are unbounded and lossless (the standard asynchronous environment). 
In both models there is no assumption about message ordering over the links, thus, no FIFO is assumed.\footnote{If FIFO is assumed, the problem is trivially solvable for the unbounded link capacity variant. The lower bound for the bounded link capacity variant can be proved for FIFO links as well.
} 
We  show a common suffix impossibility results for both models. 

%
%

We also study the problem in a system that is subjected to transient faults.
For the self-stabilization model, we prove that a common suffix can be guaranteed for a stream of a single input value that repeats itself, even when facing any number of crash faults. 
The solution makes use of a link layer, inspired by \cite{StabilizingDataLink,StabilizingEndToEndCommunication}.
Our link layer (as we will explain later) guarantees that the number of ghost messages (messages that weren't really sent) that might be received, following the last transient fault, is at most 3. The solution does not exchange acknowledge messages. 
The lower bound technique we developed for proving the previous lower bounds  is extended  to prove that 
a common suffix cannot be guaranteed even for a stream of a single value in the presence of  a single Byzantine party. Thus, this problem shows a clear {\bf separation} between crash and Byzantine faults.

It is important to point out that the the problem discussed in this article is not equivalent to Total Ordering Reliable Broadcast. In this paper we consider a single sender in opposed to total ordering that considers multiple senders and tries to achieve a total order among all concurrent sending of messages. In addition the impossibility results of this paper are not derived from the $FLP$ impossibility (\cite{FLP}), since the single sender message sequence is solvable with infinite memory.  

We are not the first to notice that sequential (repeated) Reliable Broadcast is unsolvable under the bounded model. 
%
\cite{Ricciardi96impossibilityofRepeatedReliableBroadcast,WithFiniteMemoryConsensusIsEasierThanReliableBroadcast} also show relevant  impossibility result, though the model assumes message loss. 
Here we show that message loss is not the source of the impossibility result and the same impossibility result can be achieved without assuming messages loss.  
In addition the impossibility result in previous papers is based on the observation  that a party cannot generate a message that has been lost from the internal memory of every party. 
But this does not necessarily imply that the sequences received by different parties differ.
For example, if the input sequence is periodic, a party that lost a message may still receive the same sequence, or the same suffix, it only needs to identify when the period starts. 
Moreover, the requirement that the input sequence will be generated by an external entity is not enough, the external entity may enter the same input to each party, so it is trivial to guarantee an identical sequence of received messages.
\cite{ReliableBroadcastInSynchronousAndAsynchronousEnvironments} 
also mentioned that the repeated Reliable Broadcast problem is unsolvable when considering finite memory at each party. 

There are many distributed problems that make use of Reliable Broadcast, Tal Rabin and Ran Canetti used a version of Reliable Broadcast in their paper Fast Asynchronous Byzantine Agreement with Optimal Resilience~\cite{FastAsynchronousByzantineAgreementwithOptimalResilience}. They called this version of Reliable Broadcast A-CAST protocol. They use this protocol to achieve an asynchronous secret sharing  which leads to asynchronous Byzantine agreement.

In~\cite{ByzantinePaxos}  a weaker version of Reliable Broadcast is used to replace digital signatures. In that  version, there is a specific receiver that needs to be able to know that enough other parties received the sent messages.
In~\cite{Abraham_optimalresilience}  a version of repeated Reliable Broadcast is used to achieve approximate agreement with optimal resilience. 
In both articles, the repeated Reliable Broadcast protocol makes use of round numbers in the messages. This technique will clearly not hold in the bounded model.
Can  a bounded time stamp technique (\cite{DolevShavit:1989,BoundedConcurrentTime-StampSystemsAreConstructible,
Dwork:1999}) be used to overcome this limitation?

We concentrate on the 
bounded memory and bounded links model when considering self-stabilization, since \cite{self-stabilizationofdynamicsystems,Stabilizingcommunicationprotocols} proved that when considering unbounded links, the construction of self-stabilizing links requires an unbounded memory. 

In \cite{SelfStabilizePaxos,Dolev2015, WhenConsensusMeetsSelfStabilization} the practically self-stabilize concept is introduced. 
A practically stabilizing protocol assumes that most protocols such as Paxos, under realistic operation have a lifetime that could not lead their counter to exceed a very high maximal value, e.g., of $2^{64}$. This can only take place in the case where a transient fault occurs. Thus, a practically infinite run is a run that lasts a long enough number of successive steps, \ie, $2^{64}$ steps. 
A practically self-stabilize protocol needs to achieve this exact behavior and be able to reach the long enough run from any initial configuration. 
In the current paper, we consider the classical (strong) self-stabilization and require that the stabilized run will last forever.

\section{The Model and Problem Definition}
Let $\Pi = [p_1, ...,  p_n]$ be a set of parties
and let $p_{sender}$ be a special party called the sender. 
An {\em honest party} follows the protocol's instructions, a crashed party follows the protocol's instructions until it crashes, and an adversary controls all Byzantine parties and instructs them what to do,  regardless of the protocol's instructions. 
Parties communicate via message passing.
There is a finite set of possible messages that may be sent, thus one cannot use infinite counters, or damp all past history in a message.   
We assume a standard asynchronous environment 
with fully connected network graph, \ie there are two directional links between any two parties. 
Pending messages over an incoming link may arrive at arbitrary order.
$p_{sender}$ has an external input stream containing values that the sender needs to broadcast to all receivers.  
The values are generated by an external source and this external source is not subjected to transient faults.

An {\em Internal State} of a party is the values of all internal variables.
A {\em Configuration}, $C$, is the internal states of all parties. 
The {\em Network State}, $N$, contains all the links and the sets of messages in them.
The {\em System State} is $S=(C,N)$.

A {\em step} is a function from one system state to another. A step  chooses a party,  
according to the party's current internal state it
potentially receives a pending message from one of its incoming links (if the chosen party is $p_{sender}$, it may read the next value from the input stream). As a function of the specific message obtained (if any) it performs the  instructed task, potentially sends messages to other parties and returns a new internal state. 
A step may also contain an action to deliver (output) some value, in this case, a single output is delivered in a step. 
Notice that not all actions can be taken at all system states. 
A step $s$ is {\em feasible} at internal state $C$ and network state $N$ if $s$ is the next step in $C$ and $s$ can be performed at $N$, \ie if the message (or input)  received in $s$ is  in an incoming link (or at the top of the input stream) of the party. 

When a party  sends a message to another party, the message is added to the set of pending messages over the communication link connecting them. 
If the link is bounded and exceeds its maximal capacity, then an arbitrary message (either the new one or a pending one) is lost. When a party tries to receive a message it arbitrarily receives one of the messages present on one of its incoming links or none (if no message is being received), and the received message is  removed from the respective link.

A {\em run} is a sequence of  configurations and steps, $[c_1, s_1, ...,  s_{n-1}, c_n]$ (could be infinite). 
A run specifies only the sequence of configurations and the steps that are performed. It does not specify the network states. Once a run is applied to a given network state it implies a sequence (could be infinite) of System States. 
A run is feasible at a network state $N$, if for each $i$, $s_i$ is feasible at $(c_i,N_i)$, where $N_i$ is the result of applying the prefix $[c_1, s_1, ...,  s_{i-1}]$ to $N$, and if $c_{i+1}$ is the result of applying $s_i$ on $c_i$. 
The {\em initial configuration} is $c_1$. A partial run is a run that starts at some intermediate configuration of some run and includes some sequence of consecutive steps of that run. 

Consider two models. 
The {\em semi-bounded model}, in which the memory is bounded and the  links are reliable and unbounded.
The fully bounded model, in which the memory is bounded and the links are unreliable, \ie messages may get lost and the number of messages in transit over a link never exceeds some constant $\bar c$. 
In the case of self-stabilization, as a result of a transient fault, the link may contain fake messages, 
The only way this may happened is because of transient faults.
Thus, there can be at most $\bar c$ fake messages that may arrive on that link. 
To eliminate adversarial link scheduling,  assume that if a message  is sent infinitely often it will arrive infinitely often. 

\begin{definition}[Suffix]
Let $a,b$ be two sequences (could be infinite). We say that $a$ is a suffix of $b$ if there is a location $i$ in b such that the sub-sequence of $b$ that starts from $i$ is the sequence $a$.
\end{definition}

The traditional {\em Reliable Broadcast} problem (cf. \cite{Chang:1984,CommunicationandAgreementAbstractionsforFault-TolerantAsynchronousDistributedSystems}) focuses on a message delivery consistency. In the current paper we are interested only in a suffix consistency.
%
%
The {\em Suffix Reliable Broadcast} (SuRB) problem is to implement a protocol that satisfies: $p_{sender}$ sends a sequence of messages, subjected to:
\beginsmall{enumerate}

\item[S1)] If an honest party delivers an infinite sequence of messages,  every honest party delivers an infinite sequence of messages and all these sequences share a non-trivial common suffix. 
\item [S2)] If $p_{sender}$ is honest and 
broadcasts an infinite sequence of messages, all honest parties  share a non-trivial common suffix with the input sequence of $p_{sender}$.

\end{enumerate}

\section{Suffix Reliable Broadcast with Unbounded Lossless Links}
We  examine first the Suffix Reliable Broadcast in the semi-bounded model, thus links are lossless.
The adversary may prevent a message from reaching its destination for as long as it wants, but eventually, all messages must reach their destinations. We show impossibility of {\em SuRB}, \ie there is no protocol that can satisfy SuRB for all input sequences in the standard asynchronous model
in the presence of a single crash fault.
The proof technique is used later in the paper to prove the lower bound results for the transient-fault model.

\commentout{
\ms{this hole part is a little tricky and "invite" more comments from the reviewers}
Paxos uses \ms{to create a replicated state machine. This technique uses an integer counters to specify the proposal number (for Paxos) and the consensus instance (for replicated state machine) (comment 2 reviewer 2)}. It is clear that in the bounded model those integers may reach their maximal value and so the \ms{replicated state machine} will stop functioning correctly. \ms{This claim is not enough to say that we cannot create a protocol that achieve the same behaviure in our bounded model as Paxos does in the unbounded}. One may try to replace the integer counters with bounded timestamps ( such as~\cite{DolevShavit:1989,BoundedConcurrentTime-StampSystemsAreConstructible} 
or~\cite{SimpleandEfficientBoundedConcurrentTimestamping,Dwork:1999}\ms{those are not self-stabilize (comment 2 reviewer 3)}) and achieve the same desired behaviour. The impossibility result we prove implies, among other things, that it is still impossible  to create a \ms{replicated state machine} in the bounded model even when only one party is prone to crash.
}

We  assume that parties do not have any prior knowledge 
regarding the input stream to be provided to the sender.

Since  communication links are unbounded we need a compact representation of the network state.
Let $\M=\{m_1,...,m_{\tiny M}\}$ be the set of possible messages, which we assume is bounded. 
The total number of links in the system is $\bar\ell=n(n-1).$
Define the {\em network matrix}, $N$, to be a matrix of size $M \times\bar\ell$, where  $N[i,j]$  is the number of $m_i$ messages  currently in link $j$. 
For two network matrices, $N_1,N_2$,  $N_1 \leq N_2$ if for each entry $i,j$, $N_1[i,j] \leq N_2[i,j]$.
Thus, the system state can be represented as $(C,N)$, where  $C$ is a configuration and $N$ is a network matrix.

Let $R_1 = [c^{(1)}_1, s^{(1)}_1, ...,  c^{(1)}_n]$ be a run feasible at state $N^{(1)}$ 
and $R_2 = [c^{(2)}_1, s^{(2)}_1, ..., c^{(2)}_n]$ at $N^{(2)}$.  
Let $ N^{(1)}_n=N^{(1)}R_1$ be the resulting network state of applying $R_1$ on $N^{(1)}$.
Assume that $c^{(1)}_n = c^{(2)}_1$
and $N^{(2)} \leq N^{(1)}_n$. 
It is possible to concatenate $R_1$ and $R_2$ and receive a new partial run 
$R = [c^{(1)}_1, s^{(1)}_1, ..., c^{(2)}_1, s^{(2)}_1, ..., c^{(2)}_n]$  feasible at state $S_1=(c^{(1)}_1,N^{(1)})$. 
This concatenation is possible, because if a run is feasible from a given network matrix then the same run is also feasible at any greater or equal network matrix, since all the messages used during the run also exist in any greater or equal network matrix.
Notice that $R_2$ could be infinite only if $N^{(1)}_n = N^{(2)}$,  since all messages need to eventually be delivered, so we need at some point to deliver messages that has not been delivered within $R_2$.

The following is the main result of this section. The full proof of this theorem appears in \sectionref{surb-imp}, and other missing proofs appear in the appendix, \sectionref{sec:unbounded}. 
\def\thmsurbtxt{
In the semi-bounded model with potentially a single crash fault, there is no protocol that satisfies the SuRB properties for all input sequences.}
\begin{theorem}[SuRB impossibility with one crash]\label{thm21}
\thmsurbtxt
\end{theorem}

The following discussions are about the main elements of the proof of  \theoremref{thm21} for the semi-bounded model. The similar impossibility result for the fully-bounded model appears in the Appendix.

\subsection{An Infinite Chain Existence}


\begin{definition}[Partial ordered Sequence]
Let $Set_S=\{a_1, ... , a_n\}$ be a partial ordered set of elements (could be infinite) and let '$\leq$' be the binary relation among elements of $Set_S$. A sequence 
$S = a_1, a_2 ... $ of the elements is called  a {\em Partial Ordered Sequence}.
$S$ is called  a {\em Fully Ordered Sequence} (or a Chain) if for every $i< j$, $a_i\leq a_j$.
\end{definition}


\begin{definition}[Chain in a Sequence]
Let $S = a_1, a_2 ... $ be a Partial Ordered Sequence (could be infinite). Let $\sigma = t_1,t_2 ... $ be an increasing sequence of natural  numbers (could be infinite).  Let $S'' = a_{t_1}, a_{t_2} ...  $ be a  sequence over some of the elements of $S$. We say that $S''$ is a {\em Chain in $S$} if $S''$ is a Fully Ordered Sequence.
\end{definition}

\begin{definition}[ Antichain in a Sequence]
Let $S = a_1, a_2 ... $ be a Partial Ordered Sequence (could be infinite). Let $\sigma = t_1,t_2 ... $ be an increasing sequence of natural  numbers (could be infinite).  Let $S'' = a_{t_1}, a_{t_2} ...  $ be a  sequence over some of the elements of $S$. We say that $S''$ is a {\em Antichain in $S$} if each two elements in $S''$ are not comparable.
\end{definition}

The new presentation of the network state enables us to talk about properties regarding sequences of network states. It is easy to see that a sequence of Network Matrices composes a partial order sequence. Dilworth's lemma~\cite{IntroductionToCombinatoric} states that each infinite Partial Order set must contain an infinite chain or an infinite Antichain. Unfortunately, the lemma does not refer to sequences. In order to use Dilworth's lemma we must first expand it to sequences:

\def\DilworthsLemmaSeq{
A Partial Ordered Sequence with infinite number of elements must have an infinite Chain or an infinite Antichain.
}

\begin{lemma}[Dilworth's lemma for sequences (infinite version)]\label{Lemma14}
\DilworthsLemmaSeq
\end{lemma}

A short intuition regarding the proof: Let us assume we have a partial order sequence. It is always possible to match each element with its corresponding index (its location in the sequence). This way we define a partial order set. Thanks to Dilworth's lemma we know that this set must have a Chain or an Antichain. Let us assume it contains a chain. The elements composing the chain must also compose a chain in the original sequence. Similar arguments hold for an Antichain.

Looking at an infinite sequence of Network Matrices we can prove the following lemma:

\def\lemchainexists{
Let $S$ be an infinite sequence of network matrices. 
$S$ must contain an infinite Chain.
}
\begin{lemma}[An infinite Chain Existence]\label{Lemma16}
\lemchainexists
\end{lemma}

The main idea is to show that there is no infinite Antichain which directly leads to the inevitable conclusion (using Lemma \ref{Lemma14}), that there must be an infinite chain. We used reduction to show that there is no infinite Antichain. 
We look at the matrix as a long vector. The case for $1 \times 1$ vectors is trivial, since the vectors cells are natural numbers.
Assume correctness on $n \times 1$ vectors and prove for $(n+1) \times 1$ vectors. Assume to the contrary that there is an infinite Antichain. Let us look at the first $n$ cells of this infinite Antichain, by the induction hypothesis those cells do not contain an infinite Antichain so they must contain an infinite chain. Looking at the element of this chain with the minimal value in the $n + 1$ cell, call it $min_{element}$. The next element in the chain that comes right after $min_{element}$, must be greater or equal than $min_{element}$. So we found two elements located in this infinite Antichain that are comparable. This contradicts the assumption that this is actually an infinite Antichain.

\subsection{SuRB Impossibility}\label{surb-imp}
We define the problematic sequence which we use to prove the impossibility result:

\begin{definition}[Incremental Sequence of Messages]
$M_{x,{\ell-repeated}}$ is a sequence $M = x, ...,  x$  that contain only message $x$ for $\ell$ times.
Let $M = {x_1}, ..., {x_n}, ...$ be an infinite sequence of messages, where $x_k \in \{0,1\}$. We say that $M$ is an {\em Incremental Sequence} of messages if for every sub-sequence $M_{x,{\ell-repeated}}$ that ends at some position $i$ in $M$ there exists $\ell' > \ell$ and a sub-sequence: $M_{(1-x),{\ell'-repeated}}$ that starts at some position $j$, where $j>  i$.
\end{definition}

This sequence contains an increasing number of consecutive ones and zeros. This sequence is problematic because for each sub-sequence of such sequence we can always find a later subsequence such that those two sub-sequences do not share a common divider. Using Lemma \ref{Lemma19} (appeared in the Appendix) we conclude that if there is no common divider, by switching these two subsequences we obtain a new and different sequence.


Now assume that there is a protocol that achieves the SuRB properties. The following lemma must hold for such a protocol (the full proof is in  the appendix, \sectionref{sec:unbounded}):

\def\lemlegalsys{
Let $PR$ be a protocol that satisfies the SuRB properties even when one party is prone to crash. Let $S = (C_S,N_S)$ be a reachable System State. There is a partial run $R$ of $PR$ such that $R$ starts from $C_S$, feasible at $N_S$ and in $R$ there are two honest parties who do not deliver the same sequence of messages.
}
\begin{lemma}\label{Lemma23}
\lemlegalsys
\end{lemma}

To obtain some intuition regarding the proof of the lemma  assume that when the protocol reaches $S$ the sequence given to $p_{sender}$ via the input stream is an incremental sequence, $inc-seq$. 
Assume that some party $p$ crashes (or dormant) and takes no steps. 
There must be a run $R$ that starts at $C_{S}$ feasible at $N_S$ and in $R$ each honest party delivers a suffix of $inc-seq$. 
Now wait until we reach the System State, $S_{start}$, in which each party starts delivering the common suffix (except $p$ who crashed). Look at the sequence of System States starting right after $S_{start}$. 
Looking only at the Network Matrices of this sequence we gain an infinite sequence of Network Matrices. This infinite sequence must contain an infinite chain ( \lemmaref{Lemma16}). 
Looking at the configuration part of the infinite chain there must be a configuration that repeats itself infinitely many times (the memory is bounded so there is a finite number of possible configurations). 
We obtain an infinite sequence of System States located at $R$ in which the Network Matrices are increasing and the Configurations repeat. Call this sequence $repeated-seq$. Call the first, and the second elements of $repeated-seq$: $repeated-seq_1$, and $repeated-seq_2$ and call the partial run that starts at $repeated-seq_1$ and ends at $repeated-seq_2$: $repeated-seq_{1 - 2}$. 
Now we can take $repeated-seq_{1 - 2}$ and repeat it again from $repeated-seq_2$ (because the Network Matrix is larger and the Configurations are the same). 
Call this new run $R_1$. 
We can also repeat $repeated-seq_{1 - 2}$ from any system state in $repeated-seq$. 
Let us choose to repeat it when we reach a System State at which the sequence given to $p_{sender}$ in the partial run that starts at $repeated-seq_{2}$ does not share a common divider with the input sequence given to $p_{sender}$ during $repeated-seq_{1 - 2}$. 
Call this new run $R_2$. Notice that we build two runs that start at $C_S$ feasible at $N_S$ and in both of them the sequences that are given to $p_{sender}$ are different (\lemmaref{Lemma19}). 
In addition both runs are composed of the same partial runs so they both reach the same System State (we formally prove this in the Appendix). The crashed party $p$ who starts taking steps at the common configuration cannot distinguish between the two runs and cannot know which sequence to deliver.

\begin{proof}[Proof of \theoremref{thm21}]
%
Assume that there is a protocol that satisfies the SuRB properties even when one party is prone to crash. Let $PR$ be such a protocol. Let $S_{init} = (C_{init},N_{init})$ be the initial System State of $PR$. 
By  \lemmaref{Lemma23}, there is a run $R$ starting from $C_{init}$ and feasible at $N_{init}$ such that when applying $R$ on $N_{init}$ there are two honest parties, $p_1$ and $p_2$, that do not deliver the same sequence of messages. 
If the changes in the delivered sequences between $p_1$ and $p_2$ is infinite we are done (since there is no common suffix), otherwise  look at  System State, $S_1 = (C_1,N_1)$, in which $p_1$ and $p_2$ start delivering the same sequence of messages. Look at system state $S_2 = (C_2,N_2)$ in which all messages that were present in the link in $N_1$ already reached their destination (at some point all messages must reach their destination). We know that $S_2$ is a legal System State because we have a partial run that starts at the initial configuration, feasible at the initial Network State and reaches $S_2$. 
Now applying again  \lemmaref{Lemma23} on $S_2$ we receive another partial run, $R_1$, that starts at the $C_2$, feasible at $N_2$ and contains two honest parties that do not deliver the same sequence of messages. 

We can repeat this argument. 
If  we reach a point in which there are two parties that do not share a common suffix then we are done, otherwise we actually construct an infinite run in which there is always a pair of parties that do not deliver the same sequence of messages and all messages eventually had reached their destinations i.e. this is a legal run in which there is no common suffix.
\end{proof}

\section{Self-stabilized SuRB of a Single Repeated Message}
Self-stabilization is an important objective that robust systems should satisfy. 
It gives the system the ability to overcome transient faults and continue functioning even when the most unexpected transient fault takes place. 
A self-stabilized system does not assume any initial system configuration, and a proof of correctness is proving  that whatever the initial configuration is, the system converges to perform the desired behavior. It is assumed, though, that the external source generating the values given to the sender via the input stream is not subjected to transient faults.

The impossibility results of the previous section clearly hold also for the self-stabilizing case. 
Therefore, the question stands on 
what are the extra assumptions under which the SuRB properties may hold. 
It turns out that if the input stream is composed of a single value being repeated forever SuRB can be satisfied, even when facing transient faults.


In the fully-bounded model we show that when assuming that the input stream is a stream of a single message repeated forever, one can guarantee a suffix reliable broadcast with any number of crash faults. But there is no solution in the presence of even a single Byzantine fault, \ie this shows a clear separation between Byzantine to crash faults.

We decided to investigate the case of a single repeated message because when considering self-stabilization the sender needs to repeatedly broadcast the same message. If the sender does not follow this approach the system may start in a configuration in which each honest party believes that it already received the broadcasted message, in this case the right message may never be received by the honest parties.

\subsection{Self-stabilized SuRB with Crash Faults}
\vspace{-1em}
Before describing the protocol we  build a  link layer abstraction that satisfies some basic properties. The idea is inspired by~\cite{StabilizingDataLink} and 
can be expanded to a more complex model as described in~\cite{StabilizingEndToEndCommunication}. 
This link layer will allow us to bound the convergence time of the protocol.

\ \\
\noindent{\bf Link-Layer Definitions:}
\beginsmall{enumerate}
\vspace{-1em}
\item \textsl{Lost message:}
a message $m$ that was sent by party $P_{1}$ to party $P_{2}$ is called a `Lost Message' if $m$ was sent by $P_{1}$ but will never be received by $P_{2}$.

\item \textsl{Ghost message;}
a message $m$ that was sent by party $P_{1}$ to party $P_{2}$ is called a `Ghost Message' if $m$ was received by $P_{2}$ but was never sent by $P_{1}$.

\item \textsl{Duplicated message:}
a message $m$ that was sent by party $P_{1}$ to party $P_{2}$ is called a `Duplicated Message' if $m$ was sent by $P_{1}$ and arrived to $P_{2}$ more than once.

\item \textsl{Reorder Messages:}
messages $m_{1}$ and $m_{2}$ that were sent by party $P_{1}$ to party $P_{2}$ are called `Reordered Messages' if $m_{1}$ was sent by $P_{1}$ before $m_{2}$ but $P_{2}$ received $m_{2}$ before $m_{1}$.

\end{enumerate}


\begin{definition}[($\alpha,\beta,\gamma,\lambda$)-Stabilizing Data-Link]
We call a communication abstraction between two parties $P_{1}$ and $P_{2}$, where $P_{1}$ is the sender and $P_{2}$ is the receiver, an ($\alpha,\beta,\gamma,\lambda$)-Stabilizing Data-Link  if the following condition holds:
\beginsmall{enumerate}

\item If $P_{1}$ is honest- the sending operation eventually terminates.
\item If $P_{1}$ and $P_{2}$ are honest then:
	\beginsmall{enumerate}
		\item Only the first $\alpha$ messages after the last transient fault may be lost.
		\item Only the first $\beta$ messages after the last transient fault may be duplicated.
		\item Only the first $\gamma$ messages after the last transient fault may be ghost.
		\item Only the first $\lambda$ messages after the last transient fault may be reordered.
	\end{enumerate}
\item If $P_{1}$ and $P_{2}$ are honest and $P_{1}$ sends $m$ infinitely many times in a row then $P_{2}$ will deliver $m$ infinitely many times.

\end{enumerate}
\end{definition}


We  show a protocol that provides an ($\infty$, 0, 3, 0)-Stabilizing Data-Link, The protocol ensures  that at most {\bf three} ghost messages will be received following  a transient fault. The main idea in the protocol is  that a message that has been received $\bar c + 1$ times cannot be a ghost message (since there can only be $\bar c$ ghost messages in a link following a transient fault). 
So if the receiver receives a message $\bar c + 1$ times it delivers this message. 
To enable the receiver to receive the sent message $\bar c + 1$ times the sender sends each message $\bar c + 1$ times. 
Notice that we do not claim to achieve any reliability at this layer, therefore, no ACK messages are used. 
An infinite number of messages may still get lost. 
One last observation that we need to mention is that in \cite{StabilizingDataLink} it has been proven that there is no protocol that can guarantee zero duplicated messages. 
Here we claim to achieve such a protocol (with zero duplicated messages). 
The impossibility result obtains in \cite{StabilizingDataLink} assuming the duplication occur because of a ghost message. 
Here we simply count these messages as ghost messages and not as duplicate messages. 
It is simply a matter of interpretation that does not affect the protocol.  


\ \\
\noindent{\bf The Link Layer Protocol:}

%
%
%

\begin{algorithm}[!ht]
\footnotesize
\SetNlSty{textbf}{}{:}
 \setcounter{AlgoLine}{0}
\begin{tabular}{ r l }
  &   {\bf Function} Send($msg$)\\
  & \tb {\bf For} $i=0$ to $\bar c+1$ {\bf do}\hspace{0.5in}\hfill\textit{/* $\bar c$ is the link capacity */}\\
  & \tb\tb {\bf Send} $msg$;\\
  &\tb {\bf End}.
\end{tabular}
 \caption{The Link-Layer sender pseudocode}\label{alg:sender-link}
\end{algorithm}

\begin{algorithm}[!ht]
\footnotesize
\SetNlSty{textbf}{}{:}
 \setcounter{AlgoLine}{0}
\begin{tabular}{ r l }
  &   {\bf Data} $LastMessage$, $Counter$;\\
  &   {\bf On Receiving} $msg$;\\
  & \tb {\bf if} $msg = LastMessage$ {\bf then}\\
  & \tb\due $Counter++$;\\
  & \tb\due {\bf if} $Counter \geq \bar c+1$ {\bf then}\hspace{0.5in}\hfill\textit{/* $\bar c$ the link capacity  */}\\
   & \due\due deliver $msg$;\\
   & \due\due $LastMessage := null$;\\
    &\due\due $Counter := 0$;\\
  & \due {\bf else} $LastMessage:= msg$; \\
   &\tb\due $Counter := 1$;\\
  & \tb {\bf End}.
\end{tabular}
 \caption{The Link-Layer receiver pseudocode}\label{alg:receiver-link}
\end{algorithm}
%
%
%
%
%
%

\algref{alg:sender-link} is the link-layer sender algorithm. 
The sender simply sends the given message $\bar c + 1$ times and returns.
\algref{alg:receiver-link} is the link-layer receiver Algorithm. The algorithm maintains two variables: $LastMessage$ is the last message the receiver received from the sender, and $Counter$ that counts the number of times the receiver received $LastMessage$.
When receiving a message from the sender the receiver checks
whether the new message is the same as  $LastMessage$, if so it increases the $Counter$ by one and if the $Counter$ reaches $\bar c + 1$ it delivers the message, otherwise it sets $LastMessage$ to be the new message and resets the $Counter$ to 1.

\def\thmlinklayer{
In the fully bounded model, where the bound on any link capacity is $\bar c$, the Link Layer protocol (\algref{alg:sender-link}, \algref{alg:receiver-link}) satisfies the ($\infty$, 0, 3, 0)-Stabilizing Data-Link specifications.}
\begin{theorem}\label{thm:linklayer}
\thmlinklayer
\end{theorem}
\begin{proof}
\begin{claim}[Link Layer Claim 1 - Sending termination (condition 1)]
If $P_{1}$ is honest then the sending function eventually terminates.
\end{claim}

The send function is sending the given message $\bar c + 1$ times. So eventually it terminates.

\begin{claim}[Link Layer Claim 2 - Non-triviality (condition 3)]
If $P_{1}$ and $P_{2}$ are honest and $P_{1}$ sends $m$  infinitely many times in a row then $P_{2}$ will deliver $m$ infinitely many times.
\end{claim}

From the link properties we know that if $P_{1}$ sends $m$ infinitely many times in a row then $P_{2}$ will receive $m$ infinitely many times. Which means that $P_{2}$ will see $m$ $\bar c + 1$ times in a row infinitely many times and will deliver $m$ infinitely many times.

\begin{claim}[Link Layer Claim 3 - No duplication (condition 2.b)]
Let $P_{1}$ and $P_{2}$ be two honest parties. Assume $P_{1}$ sends $m$ to $P_{2}$ $k$ times. $P_{2}$ will deliver $m$ no more than $k$ times (unless $m$ was delivered by $P_{2}$ as a ghost message but this can happen no more than 4 times as we will see later).
\end{claim}

If $P_{1}$ sends $m$ to $P_{2}$ $k$ times  then the message is  actually being sent $m$ $(\bar c + 1)k$ times. $P_{2}$ delivers $m$ when it sees it $\bar c + 1$ times in a row. The communication layer does not duplicate messages so $P_{2}$ will receive $m$ no more than $(\bar c + 1)k$ times and will deliver $m$ no more than $k$ times.

\begin{claim}[Link Layer Claim 4 - No reordering (condition 2.d)]
Let $P_{1}$ and $P_{2}$ be two honest parties and let $m_{1}$ and $m_{2}$ be messages sent by $P_{1}$ to $P_{2}$. If $P_{1}$ sends $m_{1}$ before $m_{2}$ then $P_{2}$ will not deliver $m_{2}$ before $m_{1}$.
\end{claim}

Assume that $P_{1}$ sends $m_{2}$ after it sends $m_{1}$. Once $P_{1}$ starts sending $m_{2}$, it stops sending $m_{1}$ so when $P_{1}$ starts sending $m_{2}$ there is at most $\bar c$ $m_{1}$ messages that may arrive to $P_{2}$ after $m_{2}$. So when $m_{2}$ first arrives to $P_{2}$ there is not enough messages left in the communication link for $m_{1}$ to be delivered by $P_{2}$. So if $m_{1}$ was delivered by $P_{2}$ it must be before $m_{2}$ is delivered.

\begin{claim}[Link Layer Claim 5 - At most three ghost messages (condition 2.c)]
Let $P_{1}$ and $P_{2}$ be two honest parties and let $m_i (i = 1, ... ,  4)$ be four messages delivered after the last transient fault by $P_{2}$ (in this order), then $m_4$ must be a real message that was sent by $P_{1}$.
\end{claim}

\noindent\textbf{First Ghost Message:}
Because of the transient fault $P_{2}$ can start believing that it has to deliver $m_{1}$, so  $m_{1}$ may be fake.\\
\\
\textbf{Second Ghost Message:}
Because of transient fault $P_{2}$ can start the run when its $LastMessage$ are set to $m_{2}$ and its counter is set to 1 (or higher). Also there could be $\bar c$ $m_{2}$ messages in the communication link. So $m_{2}$ may also be fake.\\
\\
\textbf{Third Ghost Message:}
Because of the transient fault $P_{1}$ can start the run inside the loop of the send function believing it is sending $m_3$ and so $m_3$ may be received by $P_{2}$. So the third message may also be fake.\\
\\ 
\textbf{The Fourth message cannot be Fake:}
Assume $m_4$ is also fake. That means that $P_{1}$ did not send $m_4$. $m_4$ must be in the system somehow before it was received by $P_{2}$. It is not in the $LastMessage$ variable of $P_{2}$ (because it was caught by $m_3$). So $m_4$ must arrive from the communication link but in this case there are at most $\bar c$ $m_4$ messages that may arrive so $m_4$ cannot be delivered by $P_{2}$. And so $m_4$ cannot be a fake message.
\end{proof}


\noindent{\bf The self-stabilized SuRB Protocol:}

%
%
%
%
%
%
%
%
%

\begin{algorithm}[!ht]
\footnotesize
\SetNlSty{textbf}{}{:}
 \setcounter{AlgoLine}{0}
\begin{tabular}{ r l }
  & {\bf While} $\true$ {\bf do}\\
    & \tb  {\bf Input} $Message$\\
  & \tb {\bf Send} $Message$ to all parties using the Link-Layer;\\
  & {\bf End}.
\end{tabular}
 \caption{$p_{sender}$ pseudo-code, $SuRB\!\_broadcast$}\label{alg:ssurb-sender}
\end{algorithm}

\begin{algorithm}[!ht]
\footnotesize
\SetNlSty{textbf}{}{:}
 \setcounter{AlgoLine}{0}
\begin{tabular}{ r l }
  &   {\bf On Receiving} from the Link-Layer $msg$ from $p_{sender}$:\\
  & \tb $SuRB\!\_deliver(msg)$.\\
\end{tabular}
 \caption{$p_i$ pseudo-code}\label{alg:ssurb-receiver}
\end{algorithm}

\def\thmsssurb{
When assuming an input stream of a single message repeated forever the self-stabilized SuRB protocol (\algref{alg:ssurb-sender}, \algref{alg:ssurb-receiver}) satisfies the SuRB properties for any number of crash fault.}

\begin{theorem}\label{thm:SuRB-crash}
\thmsssurb
\end{theorem}

\begin{proof}

We first prove the S2 property of SuRB:
If the sender is honest and the input stream contains a single message $m$ repeated forever, then each honest party must have a suffix of delivered messages that contains only $m$.
By the link layer properties, we know that at most three messages could be ghost messages. So from the fourth message on, all the messages delivered by each honest party must be real. $p_{sender}$ is honest so he sends $m$ infinitely many times. 
The link layer properties imply that eventually each honest party will receive $m$ and only $m$, which means that eventually each honest party will deliver $m$ and only $m$. In addition, after each honest party delivered three messages, at the latest, we can be sure that all following messages that will be delivered by every honest party will be $m$ and only $m$.

We now prove the S1 property.
If $p_{sender}$ is honest,  from S2 we  conclude that eventually each honest party will deliver only the message that has been given to the sender via the input stream and so all honest parties will share a common suffix. 

If $p_{sender}$ is has crashed (dishonest) then eventually each honest party will stop receiving messages from $p_{sender}$ and will stop delivering messages. This means that there is no party that delivers an infinite number of messages and so condition S1 clearly holds, since it considers only the case in which there is an honest party delivering an infinite number of messages.
\end{proof}



\subsection{ Self-Stabilize SuRB Impossibility When Facing a Byzantine Party}
In this section we will prove that even when the input stream given to the sender is a stream of a single message repeated forever, it is impossible to achieve the SuRB propertiesת assuming transient faults and a single Byzantine party. 
The SuRB impossibility proof (\theoremref{thm21}) does not directly apply here, since we do not consider all possible inputs, but rather a stream of a single message.
We  introduce few changes to the definitions we used in the SuRB impossibility proof. 
Self-stabilization does not require an initial System State and therefore any System State may be the initial System State. 
Also we need to consider the possibility of transient faults and the maliciousness of one party so now a run may contain the possibility to move from one configuration to another without applying any step. 
The only way such things may happen is because of transient faults, which cause  parties to change their internal states, or because of maliciousness of one party, which also causes it to change its internal state. 
All the missing proofs appear in the appendix, \sectionref{sec:ssurb-lower-apndx}.

%
%
%



Self-Stabilization obligates us to continuously keep delivering messages. If at some point the protocol stops delivering messages then it could be that the run would start at this point  and  no messages will ever be delivered and there is no way that any party will share a common suffix with the input stream. This  leads to the following lemma: 

\def\leminf{
Assume that the input stream given to $p_{sender}$ is a steam of a single message repeated forever. Let $PR$ be a protocol that satisfies the conditions of the  SuRB problem and is resistant to one malicious party. Let $S = (C,N)$ be an arbitrary System State. For each party $p' \in P$,  $p'\not=p_{sender}$, there is always a run $R$ that starts from $C$ and is feasible at $N$ such that $p'$ takes no steps in $R$ and each honest party delivers an infinite number of messages.
}
\begin{lemma}\label{Lemma7}
\leminf
\end{lemma}

A transient fault may cause parties to change their internal state illegally, which will cause the System State to change without any step taken. 
In addition, a Byzantine party is also able to change its internal state illegally, which will also cause the System State to change without any step taken. These observations imply the following lemma:

\def\lemconcat{
Assume that the input stream given to $p_{sender}$ is a steam of a single message repeated forever. 
Let $PR$ be a protocol that satisfies the conditions of a SuRB problem and resistant to one Byzantine party. Let $R_1 = [c^{(1)}_1,s^{(1)}_1, ...,  c^{(1)}_n]$ and $R_2 = [c^{(2)}_1,s^{(2)}_1, ...,  c^{(2)}_n]$ be two runs such that $c^{(1)}_n$ and $c^{(2)}_1$ are different only in the internal state of $p_{sender}$. 
Let $N_1$ be a network state such that $R_1$ is feasible at $N_1$ and let $N_2$ be the resulting network state of applying $R_1$ on $N_1$. Assume that $R_2$ is feasible at $N_2$. 
If $p_{sender}$ is Byzantine or a transient fault causes the last configuration of $R_1$ to be the same as the first configuration of $R_2$, then we can concatenate $R_1$ and $R_2$ and receive a new run $R = [c^{(1)}_1,s^{(1)}_1, ...,  c^{(1)}_n,c^{(2)}_1,s^{(2)}_1, ...,  c^{(2)}_n]$ feasible at $N_1$.
}
\begin{lemma}[Concatenation of runs using a transient fault or a Byzantine sender]\label{Lemma8}
\lemconcat
\end{lemma}

\lemmaref{Lemma7} and  \lemmaref{Lemma8} allow us to concatenate an infinite number of runs together and create a run in which each honest party delivers an incremental sequence.


\def\lemssinfinc{
Assume that the input stream given to $p_{sender}$ is a steam of a single message repeated forever. 
Let $PR$ be a protocol that satisfies the conditions of the  {\em SuRB} problem and resistant to one Byzantine party. 
For each party  $p' \in P$,  $p'\not=p_{sender}$, and for each System State $S = (C,N)$ there is a run $R$ that starts from $C$ and is feasible at $N$ such that $p'$ takes no steps in $R$ and each honest party delivers an infinite incremental sequence. 
Also, the only party that might change its internal state illegally in $R$ (because of transient fault or maliciousness) is $p_{sender}$.
}

\begin{lemma}\label{Lemma9}
\lemssinfinc
\end{lemma}

In the proof of \lemmaref{Lemma9} we build a run  $R_{final}$ 
that causes every honest party (other that $p'$) to deliver an
incremental sequence
by concatenating an infinite number of runs together. 
The basic idea is that we can let $p_{sender}$ sends some message, say $m_0$, repeatedly, until (using condition S2) every honest party delivers that message.  
Then, using a transient fault, cause $p_{sender}$ to start sending a different message, say $m_1$, until every honest party delivers a longer suffix of $m_1$ than it delivered $m_0$ before.
Now we can switch back to $m_0$ to deliver even a longer suffix.
Because in all those runs $p'$ takes no steps it means that $p'$ takes no steps in $R_{final}$. 
So $p'$ clearly assumed to be failed (or Byzantine), since an honest party must start, at some point, taking steps.

The only way we could concatenate those runs is by using  transient faults that occasionally change the internal state of the sender. This means that $R_{final}$ contains an infinite number of transient faults.
But it should not concern us since in the proof below we  use $R_{final}$ only to construct an adversarial behavior of a Byzantine sender, and we wake up the sleeping party $p'$.



\begin{theorem}[SuRB impossibility assuming a transient fault and a single Byzantine party]\label{thm22}
Assuming that the input stream is a stream of a single messages repeated forever. In an asynchronous system and fully bounded model, there is no self-stabilizing protocol that solves the SuRB problem in the presence of a single Byzantine party.
\end{theorem}

\begin{proof}[Proof outline:]
Let $p'$ ($p'\not=p_{sender}$) be a party in $P$ and let $S = (C,N)$ be an arbitrary system state. 
By \lemmaref{Lemma9} we know that there is a run $R$ that starts in $C$ and is feasible at $N$ such that $p'$ takes no steps in $R$ and each honest party delivers an infinite incremental sequence of messages. 


Let's look at the sequence $S_{seq} = S_1, ... , S_n , ...$ of System States we receive by applying $R$ on $N$. Let $S_{seq_1}$ be a sequence of System States from $S_{seq}$ such that the next step that follows each of the System States of $S_{seq_1}$ is that $p_{sender}$ reads an input from the input queue. 
The memory of each party is bounded and the network links are also bounded so there must be a System State that repeats itself infinitely many times during $S_{seq_1}$, let's call this System State $S_{inf}$ and let $S_{inf-seq}$ be the sequence of $S_{inf}$'s in $S_{seq}$. Notice that all the configurations in $S_{inf-seq}$ are the same and the network states are also the same.
Let $seq$ be the sequence of input values used in the construction of the run $R$ in the proof of \lemmaref{Lemma9}.

Now we can continue the same way as \lemmaref{Lemma23} and \theoremref{thm21} (SuRB impossibility in the fully-bounded model) and prove that there is a run $R_{final}$ in which there is no common suffix. The only problem is that we compose $R_{final}$ from $R$ and in $R$ there is an infinite number of transient faults, which means that $R_{final}$ does not necessarily contain a common suffix. 
Looking at the proofs of \lemmaref{Lemma23} and \theoremref{thm21} we will notice that during the proof we woke up the sleepy party, which means that now we can assume that the ``illegal" changes in the internal state of $p_{sender}$ is because of maliciousness and not because of transient faults, \ie the Byzantine sender invents these values ($seq$) whenever a value needs to be read to produce $R$. 
Now $R_{final}$ contains no transient faults, which means that it must contain a common suffix but it actually doesn't.
\end{proof}

\section{Discussion of the Impossibility Results}
\vspace{-1em}
The  SuRB impossibility  shows that self-stabilization has serious limitations even in the fully bounded model. It surfaces the main difference between Byzantine behavior and weaker fault models.
It  shows that there is no general technique that can take a Byzantine tolerant distributed protocol and stabilize it. 
Asynchrony, Byzantine and self-stabilization are necessary to prove the impossibility result and without any one of them the problem is  solvable.


The  SuRB problem does not require honest parties to deliver an infinite sequence of messages if the sender is faulty. 
Hypothetically, someone may find a protocol that satisfies the  conditions and stop delivering messages at some point. 
The  self-stabilizing SuRB impossibility result shows that no such protocol exists. 
Notice that the requirement that each honest party will deliver an infinite sequence of messages does not come from the problem definition but from the system model. This is why the self-stabilizing SuRB impossibility result  is more surprising than the impossibility result of the repeated \RR. In the \RR problem we can intuitively understand that we cannot save the delivered sequence forever,  
while in the self-stabilized SuRB problem with transient faults we only have to agree on one message. Our proof shows that we cannot agree on a single message when we consider asynchrony, self-stabilization, and Byzantine concurrently.

In~\cite{SelfStabilizePaxos}, 
a ``practically'' self-stabilizing Paxos with crashed faults is introduced.
%
It looks like this result contradicts the impossibility result of the SuRB problem. 
Our result actually reinforces the protocol presented in~\cite{SelfStabilizePaxos}. 
\cite{SelfStabilizePaxos} doesn't solve  
a regular self-stabilization,
but rather 
requires stabilized run to last only long enough for any concrete system's time-scale. 
A run is divided into epochs. Each epoch is  a stabilized run. When an epoch terminates, the history 
is  cleared and a ``blank'' history is started. 

In this article we showed that when considering the case in which the sender sends a single repeated message the SuRB problem is solvable considering transient and crash faults. Also we showed that the problem is not solvable considering all possible input sequences. It is interesting to understand what are the minimal assumptions about the input sequence in which the problem is solvable. We leave it open to draw the line that will turn the problem to be solvable. 

\commentout{
\ms{An epoch is considered infinite, it is like counting from one to $2^{64}$. It is mentioned that such counting at every 
nanosecond will take us about 500 years to finish. 
The weakness is that 
if a run starts from some configuration in which the epoch is in its ``youth'', it will take about 500 years for the epoch to terminate, which means that in this time the history will contain fake values. 
These values remain in the history because of transient faults - reviewer 3 has a problem with this sentence, I still think its correct}. 
} 






\bibliographystyle{plain}
\bibliography{mylib}

\newpage

\appendix

\noindent {\Large \bf Appendix}

\section{SuRB Impossibilities}\label{sec:unbounded}
\subsection{SuRB Impossibility with Unbounded Lossless links}
\commentout{
In this section we will examine the Sequential Reliable Broadcast when the links are not bounded and are lossless. In this model the capacity of each link is infinite and messages do not get lost. We may prevent a message from reaching its destination for as long as we want but eventually all messages must reach their destinations. We are going to show that this problem is also unsolvable in this model even when one party is prone to crash.

\subsection{Definitions}
The definitions we defined in the model and in \sectionref{Definitions} also apply here. We only point out an important insight that must be taken to consideration when considering unbounded lossless links. The System Configuration as defined in the model is the combination of the internal state of each party along with the messages in the communication links. Since the communication links are unbounded there is an infinite number of possible configurations. Let $m_1,...,m_M$ be a set of size $M$ of the possible messages in the system (we assumed that there is a finite number of possible messages in the system so $M$ is not infinite). Let $ch_1, ... ch_{CH}$ be a set of size $CH$ of all the links in the system ($CH$ is also finite because there is a finite number of links in the system, two links between every two parties). We define a network matrix, $N$, to be a matrix of size $M \times CH$ such that in the $i,j$ position of $N$ there is a number indicating how many $m_i$ messages located in $ch_j$ link. We will redefine the system state to be a tuple $(C,N)$ such that $C$ is a tuple $(c_1,...,c_n)$ of internal states of $p_1,...,p_n$, respectively (we will call it system internal state) and $N$ is a network matrix. Let $N_1$ and $N_2$ be two network matrices, we say that $N_1 \leq N_2$ if for each $0 \langle  i \leq M$, $0 \langle  j \langle  CH$, it holds that $N_1[i,j] \leq N_2[i,j]$. We expand the concatenation principle as follow: Let $R_1 = [c^{(1)}_1, s^{(1)}_1, ...,  c^{(1)}_n]$ and $R_2 = [c^{(2)}_1, s^{(2)}_1, ..., c^{(2)}_n]$ be two partial runs. Assume that $c^{(1)}_n = (C_{c^{(1)}_n},N_{c^{(1)}_n})$ and $c^{(2)}_1 = (C_{c^{(2)}_1},N_{c^{(2)}_1})$. Assume that $C_{c^{(1)}_n} = C_{c^{(2)}_1}$ and $N_{c^{(2)}_1} \leq N_{c^{(1)}_n}$. It is possible to concatenate $R_1$ and $R_2$ and receive a new partial run $R = [c^{(1)}_1, s^{(1)}_1, ..., c^{(2)}_1, s^{(2)}_1, ..., c^{(2)}_n]$. Notice that $R_2$ could be infinite only if it also true that $N_{c^{(1)}_n} \leq N_{c^{(2)}_n}$ otherwise we must at some point deliver  messages that are not used by $R_2$ and so $R_2$ must end in order for us to do that.
}

\begin{lemma}\label{Lemma17}
Let $R$ be a run and let $S_1 = (c,N_1),S_2 = (c,N_2)$ be two System State such that $R$ is feasible at $S_1$ and at $S_2$. 
Assume that after applying $R$ on $S_1$ we reach System State $S_1' = (c_1',N_1')$ and after applying $R$ on $S_2$ we reach System State $S_2' = (c_2',N_2')$. Assume that for some $i,j$, $N_1[i,j] = k_1$ and $N_1'[i,j] = k_1 + t_1$ (the $i,j$ location of $N_1$ changes by $t_1$ after applying $R$ on $S_1$, $t_1$ can be negative). Assume that $N_2[i,j] = k_2$ and $N_2'[i,j] = k_2 + t_2$ (the $i,j$ location of $N_2$ changes by $t_2$ after applying $R$ on $S_2$, $t_2$ can be negative). It must be that $t_1 = t_2$. 
\end{lemma}
\begin{proof}
The run $R$ added and removed the same number of message $m_i$ from link $j$ so after applying it on $S_2$ and on $S_1$ the difference of the number of $m_i$ messages on link $j$ must be the same.
\end{proof}

\begin{lemma}\label{Lemma18}
Let $N$ be some network state. Let $R$ be a partial run. Assume that $R$ is feasible at $N$ and let $N_1 = NR$. Assume that $N \leq N_1$. $R$ is feasible at $N_1$ and $N_2 = N_1R$ satisfies $N_1 \leq N_2$.
\end{lemma}

\begin{proof}
Since $N \leq N_1$ it is clear that  $R$ is feasible at $N_1$. 
All that is left to show is that $N_1 \leq N_2$. For each $i,j$ there is a $k \geq 0$ such that after applying $R$ on $N$ the $(i,j)$ cell of $N$ changes by $k$ ($k$ is greater than or equal to zero because $N \leq N_1$). So we conclude that $N_1[i,j] = N[i,j] + k$. Now when applying $R$ on $N_1$, by \lemmaref{Lemma17} the $(i,j)$ cell of $N_1$ changes by $k$, so $N_2[i,j] = N_1[i,j] + k$. And this is true for each valid $(i,j)$ so it must be that $N_1 \leq N_2$.
\end{proof}

\begin{lemma}\label{Lemma12}
Let $S$ be a Set of network matrices along with the binary relation '$\leq$' as defined above. $S$ is  a partial ordered set.
\end{lemma}

\begin{proof}
In order to show that $S$ is a partial order set we need to prove three properties: reflexivity, antisymmetry, transitivity.

Reflexivity: let $a$ be a network matrix. It is clear that each cell in $a$ is less than or equal to itself and so $a \leq a$.

Antisymmetry: let $a,b$ be two network matrices. Assume that $a \leq b$ and $b \leq a$. It is clear the $a[i,j] \leq b[i,j]$ and $b[i,j] \leq a[i,j]$ for each valid $i$ and $j$ ($i$ and $j$ that do not exceed  the matrix size). So it is clear that for each such $i$ and $j$, it holds that $a[i,j] = b[i,j]$ so it is clear that $a = b$.

Transitivity: let $a,b,c$ be three network matrices such that $a \leq b$ and $b \leq c$. It is clear the $a[i,j] \leq b[i,j]$ and $b[i,j] \leq a[i,j]$ for each valid $i$ and $j$ ($i$ and $j$ that do not exceed the the matrix size). So it is clear that for each such $i$ and $j$, it holds that $a[i,j] \leq c[i,j]$ and so $a \leq c$.
\end{proof}

\begin{lemma}[Dilworth's lemma (infinite version \cite{IntroductionToCombinatoric} section 3.5)]\label{Lemma13}
A partial order set with infinite number of elements must have an infinite Chain or an infinite Antichain.
\end{lemma}

\noindent\relem{Lemma14}[repeated - Dilworth's lemma for sequences (infinite version)]
\DilworthsLemmaSeq
\erelem

\begin{proof}
Let $S$ be an infinite Partial Ordered Sequence with the binary relation `$\leq$'. We will build a set $S'$ of tuples $(n,a)$ in the following way: for each $i > 0$ let $a$ be the $i$-th element in $S$. Add to $S'$ the element $(i,a)$. We define the binary relation `$\leq$' on $S'$ in the following way: let $a = (i,a'), b = (j,b')$ be two elements in $S'$, $a \leq b$ iff $i \leq j$ and $a' \leq b'$. 
First we will show that $S'$ along with the binary relation `$\leq$' is a partial order set. 

Reflexivity: let $a = (i,a')$ be an element in $S'$. It is clear that $i \leq i$ and since $S$ is a partial order it is also clear that $a' \leq a'$ so it is clear that $a \leq a$.

Antisymmetry: let $a = (i,a'), b = (j,b')$ be two elements in $S'$. Assume that $a \leq b$ and $b \leq a$. So it is clear that $i \leq j$ and $j \leq i$ so $i = j$. 
It is also clear that $a' \leq b'$ and $b' \leq a'$ so $a' = b'$. So we conclude that $a = b$.

Transitivity: let $a = (i,a'), b = (j,b'), c = (t,c')$ be three elements in $S'$. Assume that $a \leq b$ and $b \leq c$. So it is clear that $i \leq j$ and $j \leq t$ so $i \leq t$. It is also clear that $a' \leq b'$ and $b' \leq c'$ so $a' \leq c'$. So we conclude that $a \leq c$.

Now by \lemmaref{Lemma13} we get that $S'$ contains an infinite Chain or an infinite Antichain. Assume  that $S'$ contains an infinite Chain. Let $S'' = (i_1,a_1), ..., (i_l,a_l), ...$ be that chain. Let's order $S''$ by the $i$-th entry of each tuple and get only the $a$'s entries. We receive a Chain in $S$. We can do the same  on an Antichain and receive an Antichain in $S$.
\end{proof}

\begin{lemma}[An infinite Antichain Inexistence]\label{Lemma15}
Let $S$ be an infinite sequence of network matrices along with the binary relation '$\leq$' as defined above. $S$ does not contain an infinite Antichain.
\end{lemma}
\begin{proof}
We will prove the lemma by induction on the network matrix size. 
We can map the network matrices to vectors such that each cell of the network matrix is mapped to a cell on this vector. 
Vector $a$ is greater than or equal to vector $b$ iff its corresponding matrix, $a'$, is greater than or equal to the corresponding matrix of $b$. 

Assume that $S$  contains an infinite Antichain. This means that there is an infinite sequence of elements from $S$ such that each two distinct elements in this sequence are incomparable. Call this sequence $S'$. 

We prove the lemma by induction of the vector length. First, assume that the vector size is $1$ (\ie Network matrix of size $1 \times 1$). Since network matrix entries are non-negative natural number, it is clear that there is no infinite Antichain, 
$S'$. 

Now assume correctness for vectors of size $n$ and prove it on vectors of size $n + 1$. 
Assume to the contrary that such $S'$ exists.
Let's look first only at the first $n$ cells of all vectors in $S'$ (the infinite Antichain). 
From the inductive hypothesis we know that when considering only those $n$ cells, there is no infinite Antichain. 
By \lemmaref{Lemma14}, we know that each infinite Partial Ordered Sequence must contain an infinite Chain or an infinite Antichain. 
Since looking at the first $n$ cells we do not get an infinite Antichain, there is an infinite Chain $S''$ such that when looking only at the first $n$ cells of the elements in $S''$ we get a totally ordered sequence (remember that $S''$ is a sequence of elements from $S'$ and $S''$ is infinite). 
Now look at the $n + 1$-st cell of the elements of $S''$.
Call the element with the minimum value in the $n + 1$ cell $min_{element}$. The element that appears in $S''$ right after $min_{element}$ must be greater or equal $min_{element}$ in the $n+1$-st cell and in the other $n$ cells. This way we actually find two comparable elements in $S'$ which implies that $S'$ is not an infinite Antichain.
\end{proof}

%

\noindent\relem{Lemma16}[repeated]
\lemchainexists
\erelem

\begin{proof}
By \lemmaref{Lemma14} we conclude that $S$ must contain an infinite Chain or an infinite Antichain. By \lemmaref{Lemma15} we conclude that $S$ does not contain an infinite Antichain, so $S$ must contain an infinite Chain.
\end{proof}

\begin{lemma}\label{Lemma30}
Let $S = (c,N)$ be some configuration and let $R_1$ and $R_2$ be two partial runs feasible at $S$. Assume that the initial and final internal configurations of both $R_1$ and $R_2$ are the same and equal $c$. Assume that  $R_2$ is feasible at $NR_1$ and  $R_1$ is feasible at $NR_2$. 
Then $NR_1R_2\equiv NR_2R_1$.
\end{lemma}
%


\begin{proof}
We assumed that $R_1$ and $R_2$ ended in the same system internal state so $NR_1R_2$ and $NR_2R_1$ system internal states are the same and equal $c$. 
All that is left to show is that the network matrices are the same. 
Let $N[i,j]$ be the value of the $i,j$ cell of $N$. Let $N_1$ be the resulting network matrix of applying $R_1$ on $C$. 
There is a value $k$ such that $N_1[i,j] = N[i,j] + k$ ($k$ could be negative). Let $N_2$ be the network matrix of $N_1R_2$. 
There exists a value $d$ such that $N_2[i,j] = N_1[i,j] + d = N[i,j] + k + d$ ($k+d$ could be negative). Let $N_3$ be the resulting network matrix of applying $R_2$ on $C$. 
By \lemmaref{Lemma17} $R_2$ changes the $i,j$ cell of $N$ by $d$ so $N_3[i,j] = N[i,j] + d$. Let $N_4$ be the network matrix of $N_3R_1$. By \lemmaref{Lemma17} $R_1$ changes the $i,j$ cell of $N_3$ by $k$ so $N_4[i,j] = N_3[i,j] + k = N[i,j] + d + k$. This is true for each valid $i,j$ so $N_4 = N_2$.
\end{proof}


\begin{definition}
Let $Seq$ be a sequence of elements and let $q$ be a natural number. $qSeq$ is a sequence of concatenating $Seq$ to itself $q$ times.
\end{definition}

\def\lemsubseqeq{
Let $a_1$ and $a_2$ be two non-empty sequences such that the size of $a_2$ is greater than the size of $a_1$. Let's look at two compositions of $a_1$ and $a_2$, $b_1 = a_1 + a_1 + a_2$ and $b_2 = a_1 + a_2 + a_1$. $b_1 = b_2$ iff there exists a sequence $S$ of size less than or equal the size of $a_1$ and two natural numbers $n,r > 0$ such that $a_1 = nS$ and $a_2 = rS$, we call $S$ the common divider.
}
\begin{lemma}[The Common Divider Lemma]\label{Lemma19}
\lemsubseqeq
\end{lemma}


\begin{figure}[t]
\centering
\includegraphics[scale=0.6]{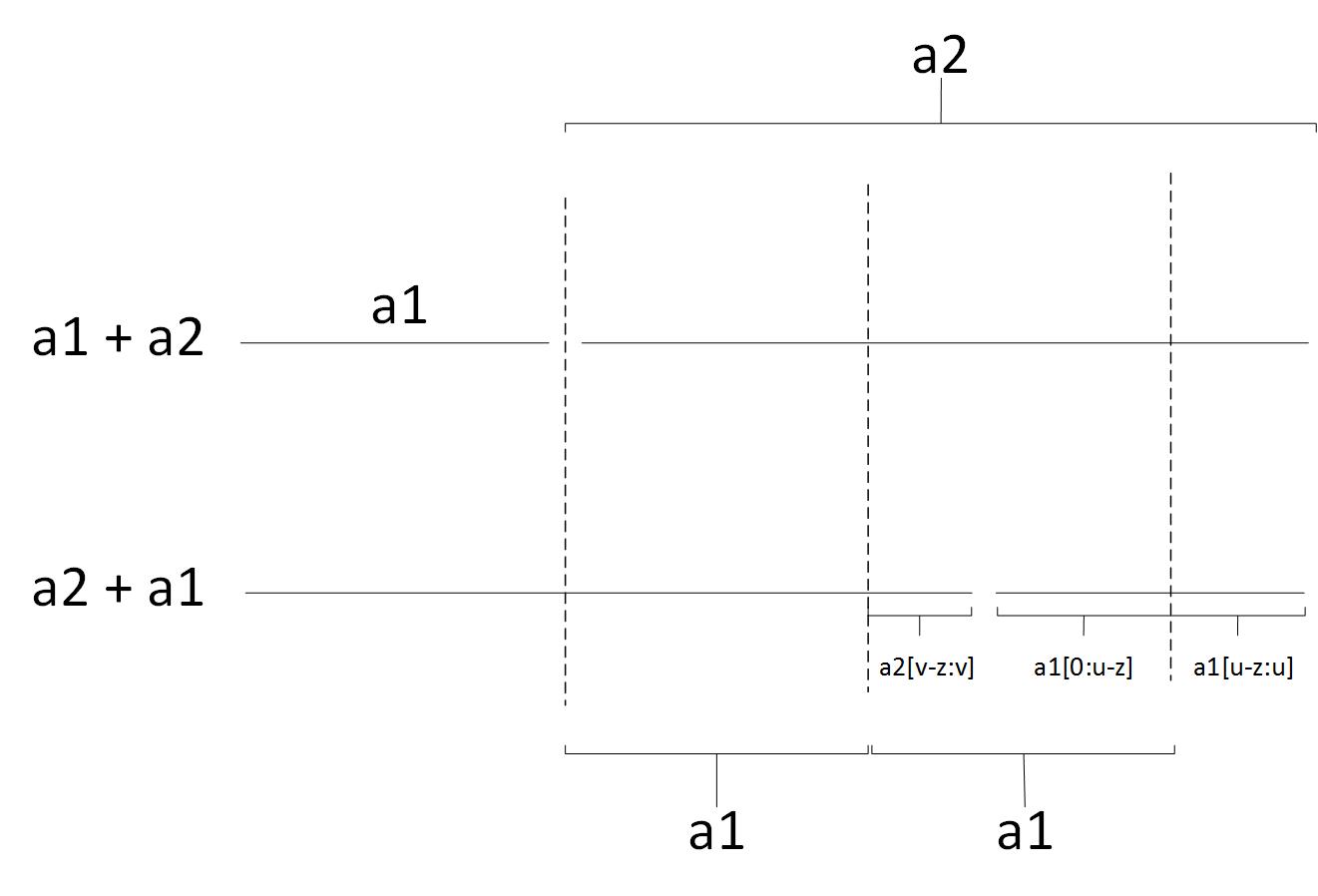}
\caption{$a_1 + a_2$ is the same as $a_2 + a_1$}
\label{fig:figure5}
\end{figure}

\begin{proof}
If $a_1 = nS$ and $a_2 = rS$ then $b_1 = a_1 + a_1 + a_2 = (2n + r)S$ and  $b_2 = a_1 + a_2 + a_1 = (2n + r)S$ so $b_1 = b_2$.
Now assume that $b_1 = b_2$. 
Sequences $b_1$ and $b_2$  are the same iff $a_1 + a_2$ is the same as $a_2 + a_1$. 
We will show that for this to happened there must be a sequence, $S$, and two natural number $n,r$ such that $a_1 = nS$ and $a_2 = rS$. We will use recursion. We first assume that the size of $a_1$ is 1 and it equals $x$. 
If there is no natural number $n > 1$ such that $a_2 = nx$ then there must be a location, $i$, in $a_2 + a_1$ such that the value in $i$ is not $x$ and the value in $i+1$ is $x$, in $a_1 + a_2$ this location moves one index foreword, so the value in the location $i+1$ of $a_1 + a_2$ and $a_2 + a_1$ is not the same and so the sequences are not the same. 
Let $u < v$ be two numbers. Assume that for each $i,j$ such that $i < u, j < v, i < j$ and the size of $a_1$ and $a_2$ is $i$ and $j$ respectively, $a_1 + a_2$ is the same as $a_2 + a_1$ if there is a sequence, $S$, and two natural number $n,r$ such that $a_1 = nS$ and $a_2 = rS$. We will prove the same result for $a_1$ of size $u$ and $a_2$ of size $v$. Let $a_2[0:u]$ be a sub-sequence of $a_2$ that starts at position 0 and ends at position $u$, it must be that $a_1 = a_2[0:u]$. 
Now for $a_1 + a_2$ to be the same as $a_2 + a_1$, $a_2[u:2u]$ must be the same as $a_2[0:u]$ which equals $a_1$. 
We can repeatedly continue until we reach the end of $a_2$, if $a_1$ fits exactly to $a_2$ we are done, otherwise there must be $z < u$ such that $a_2[v - z,v] + a_1[0,u - z] = a_1$ (notice that in such case there exists $q$ such that $a_2 = qa_1 + a_2[v - z,v]$, see figure \ref{fig:figure5}). But we also know that $a_1[u - z,u] = a_2[v - z,v]$. 
So we know that there exits $z$ such that $a_1 = a_1[0:u - z] + a_1[u - z:u] = a_1[u - z:u] + a_1[0:u - z]$ but by induction we know that for this to be true there must be a sequence $S$ and two natural numbers $l,h$ such that $a_1[u - z:u] = lS$ and $a_1[0:u - z] = hS$. So we conclude that $a_1 = (l + h)S$ and $a_2 = (q(l + h))S + lS = (q(l + h) + l)S$ and we are done.\end{proof}

\def\lemnosubseqeq{
%
Let $S$ be an incremental sequence of messages. Let $S[i,j]$ be a sub-sequence of $S$ that starts at position $i$ and ends at position $j$. There exists $t > j$ such for each $r > t$ there are no natural numbers $n,k>0$ for which there exists a sequence $Seq$ such that $S[i,j] = nSeq$ and $S[j,r] = kSeq$ \ie there is no common divider.
}
\begin{lemma}\label{Lemma20}
\lemnosubseqeq
\end{lemma}



\begin{proof}
Let's look at some $i,j$ such that $i <  j$. Assume that $S[i,j]$ contains only zeros. Since $S$ is an incremental sequence of messages,  there must be a subsequence of $S$, $S'$, that contains only ones and appears in $S$ at position $x$ after $j$. For each $x' > x$ there is no natural numbers $n,k>0$ for which there exists a sequence $Seq$ such that $S[i,j] = nSeq$ and $S[j,x'] = kSeq$. The proof for the case in which $S[i,j]$ contains only ones is the same. Now if $S[i,j]$ contains also ones and zeros. Since $S$ is an incremental sequence of messages then there must be a subsequence of $S$, $S'$, of size greater then $2(j - i)$ that contains only ones and appears in $S$ at position $x$ after $j$. Let $x'$ be a position in $S$ in which $S'$ already finished. For each $x'' > x'$ there is no natural numbers $n,k>0$ for which there exists a sequence $Seq$ such that $S[i,j] = nSeq$ and $S[j,x'] = kSeq$ because $S[j,x'']$ contains a subsequence, $S_{sub}$, of size greater than $2(j - i)$ that contains only ones. For $Seq$ to exists, $S_{sub}$ must contain it but $S[i,j]$ contains ones and zeros and $S_{sub}$ contains only ones. So if $S_{sub}$ contains $Seq$ it cannot be that $S[i,j]$ also contains it.
\end{proof}




\begin{definition}[Legal System State]
Let $PR$ be a protocol. Let $S$ be a System State. We say that $S$ is a legal System State (or reachable system state) if there is a partial run $R$ that starts from the initial configuration of $PR$, feasible at the initial network state of $PR$ and ends in $S$.
\end{definition}

\noindent\relem{Lemma23}[repeated]
\lemlegalsys
\erelem

\begin{proof}
First, we know that there is a run that starts from the initial configuration and reached $S$. 
Assume that when we reach $S$ the input sequence, $M$, of $p_{sender}$ is an infinite incremental sequence of messages. 
Let $R_S$ be the run that starts at $C_S$, feasible at $N_S$ and causes each party (except one party, $p' \neq p_{sender}$, that takes no steps in $R$, crashed) to deliver $M$. 
We know that $R_S$ exists because $S$ is legal System State so there must be a run feasible at $S$ satisfying the $SuRB$ conditions even when one party may crash and the input sequence of $p_{sender}$ is $M$. 
Let $C_{R_S}$ be a sequence of the System States obtained by applying $R_S$ on $N_S$. Let $CR$ be a sequence of System States from $C_{R_S}$ such that the next step that follows each of those System States is that $p_{sender}$ reads a value from the external input queue. $CR$ contains an infinite number of configurations. 
We know that the memory of each party is bounded so there must be at least one configuration that repeats itself infinitely many times during $CR$, let's call it $C_{repeated}$. 
Let $CR_{repeated}$ be the sequence that contains all System States of $CR$ that their configuration equals $C_{repeated}$. 
Let's look at the sequence of network matrices appearing in $CR_{repeated}$, by \lemmaref{Lemma16} we know that this sequence contains an infinite Chain. Let $CR_{repeated-chain}$ be the sequence of elements from $CR_{repeated}$ that composes the infinite Chain. $CR_{repeated-chain}$ is a sequence of the form $(C_{repeated}, N_1), ..., (C_{repeated}, N_n), ...$ such that for each $i < j$, $N_i \leq N_j$.

Let $S_{\inf-1} = (C_{repeated-1}, N_1)$ (\ie $S_{\inf-1}$ is the first System State in $CR_{repeated-chain}$). 
Let $x_1$ be an index of a message $m_1$ that is about to be read by $p_{sender}$ after $S_{\inf-1}$. 
Let $S_{\inf-2} = (C_{repeated-2}, N_2)$ be the second System State in $CR_{repeated-chain}$ and let $x_2$ be an index of a message $m_2$ that is about to be read by $p_{sender}$ after $C_{\inf-2}$. 
Let $x_3$ be an index of messages in $m_3$ such that the size of $M[x_1,x_2]$ is less than the size of $M[x_2,x_3]$. 
By \lemmaref{Lemma20}, we know that there is an index $x_4$ such that for each index $x_5 > x_4$ the sub-sequences $M[x_1,x_2]$ and $M[x_2,x_5]$ do not share a common divider. Let $S_{\inf-3} = (C_{repeated-3}, N_3)$ be the System State in $CR_{repeated-chain}$ in which $p_{sender}$ already received the message in position $Max(x_3,x_4)$ as input.

Notice that by the construction of $CR_{repeated-chain}$ the configurations of $S_{\inf-1}$, $S_{\inf-2}$, $S_{\inf-3}$ are the same and the network matrix of $S_{\inf-3}$ is greater than or equal to the network matrix of $S_{\inf-2}$, which is also greater than or equal to the network matrix of $S_{\inf-1}$.

Let's now separate $R$ to four runs. The first, $R_1$, starts at  $C_S$, feasible at $N_S$ and ends in $C_{repeated-1}$. The second, $R_2$, starts at $C_{repeated-1}$, feasible at $N_2$ and ends in $C_{repeated-2}$, the third, $R_3$, starts at $C_{repeated-2}$, feasible at $N_2$ and ends in $C_{repeated-3}$ and the last, $R_4$, starts at $C_{repeated-3}$, feasible at $N_3$ and continues infinitely with some legal run.

By \lemmaref{Lemma18}  and the fact that the network matrices are increasing, after reaching $S_{\inf-2}$ we can run $R_2$ again, reaching a new System State $S_{\inf-2-1}$ such that the network matrix of $S_{\inf-2}$ is less than or equal to the network matrix of $S_{\inf-2-1}$. Now we can run $R_3$, reaching a new System State $S_{\inf-3-1}$. To sum up we constructed a new partial run $R_5 = R_1 + R_2 + R_2 + R_3$ feasible at $N_S$.

By \lemmaref{Lemma18} and the fact that the network matrices are increasing, after reaching $S_{\inf-3}$ we can run $R_2$ again and reach a new System State $S_{\inf-3-2}$ such that the network matrix of $S_{\inf-3}$ is less than or equal to the network matrix of $S_{\inf-3-2}$. To sum up we construct a new partial run $R_6 = R_1 + R_2 + R_3 + R_2$ feasible at $N_S$.

Let $M_{5}$, $M_{6}$ be the sequence of messages delivered by applying $R_5$, $R_6$ on $N_S$ respectively. 
By \lemmaref{Lemma30} we conclude that both $R_5$ and $R_6$ ends in the same System State after applying them on $N_S$ so $S_{\inf-3-2}$ = $S_{\inf-3-1}$. Now if $p'$ wakes up in $S_{\inf-3-1}$ (or $S_{\inf-3-2}$) it cannot distinguish between $R_5$ and $R_6$. There is a scenario in which $p'$ delivers $M_{5}$ while all other parties deliver $M_{6}$ and there is a scenario in which $p'$ delivers $M_{6}$ while all other parties deliver $M_{5}$. So it is enough to show that $M_{5}$ and $M_{6}$ are different. Let $M_{R_2}$ be the sequence given to $p_{sender}$ in $R_2$ and let $M_{R_3}$ be the sequence given to $p_{sender}$ in $R_3$. For $M_{5}$ and $M_{6}$ to be the same, $M_{R_2} + M_{R_2} + M_{R_3}$ must be the same as $M_{R_2} + M_{R_3} + M_{R_2}$ (notice that by construction of $R_2$ and $R_3$, the size of $M_{R_2}$ is less than the size of $M_{R_3}$). By \lemmaref{Lemma19} we know that for $M_{R_2} + M_{R_2} + M_{R_3}$ to be the same as $M_{R_2} + M_{R_3} + M_{R_2}$, those two sequences must share a common divider. But we constructed those two sequences such that this scenario will not happen and so $M_{5}$ and $M_{6}$ must be different.
\end{proof}

\subsection{SuRB Impossibility with Bounded links}
The proof of \theoremref{thm21} uses only \lemmaref{Lemma23},
therefore in order to prove the impossibility result, it is enough to show that \lemmaref{Lemma23} holds when links are bounded.

\begin{lemma}\label{Lemma21}
Let $M$ be an infinite incremental sequence of messages. For each index $x_1$ there are two sub-sequences of $M$, $M_{0,{\ell-repeated}}$ and $M_{1,{t-repeated}}$, such that the first time those sequences appears in $M$ is after $x_1$.
\end{lemma}

\begin{proof}
Because $M$ is an infinite incremental sequence of messages we can assume that there is an infinite increasing sequence $L = \ell_1,\ell_2, ...$ such that  $M_{0,{\ell-repeated}}$ appears in $M$ for each $\ell \in L$ and there is an infinite increasing sequence $T = t_1,t_2, ...$ such that  $M_{1,{t-repeated}}$ appears in $M$ for each $t \in T$. There is only a finite number of messages that appear before $x_1$ so there is a finite number of $\ell \in L$ such that $M_{0,{\ell-repeated}}$ starts before $x_1$ and so we can always find $\ell \in L$ such that the first time that $M_{0,{\ell-repeated}}$ appears in $M$ is after $x_1$. Also there is a finite number of $t \in T$ such that $M_{1,{t-repeated}}$ starts before $x_1$ and so we can always find $t \in T$ such that the first time that $M_{1,{t-repeated}}$ appears in $M$ is after $x_1$.
\end{proof}

The following lemma is a general claim about sequences. These sequences may not reflect any possible run. When it is used later we will make sure that the sequences belong to feasible runs.

\begin{lemma}\label{Lemma22}
Let $M$ be an infinite sequence of messages.
Let $M'$ be the same sequence as $M$ without $M_{sub}$, then
$M$ is different from $M'$.
\end{lemma}

\begin{proof}
Let $x$ be the index in $M$ at which $M_{sub}$ ends.
We know that there is a sub-sequence of $M$, $m_{i,{\ell-repeated}}$, that appears in $M$ for the first time after $x$ (\lemmaref{Lemma21}). 
Let $x_1$ be the index at which $M_{i,{\ell-repeated}}$ first appears in $M$. In $M'$ we delete the sequence $M_{sub}$ which comes before $M_{i,{\ell-repeated}}$ and is not empty, so the first time $M_{i,{\ell-repeated}}$ appears in $M'$ must be before it appears in $M$. 

Let $x_2$ be the index in $M'$ at which $M_{i,{\ell-repeated}}$ first appears in $M'$ (remember that by definition $M_{i,{\ell-repeated}}$ contains only messages $m_i$). There must be an index $x_3$ such that $ x_2 \leq x_3 \leq x_2 + \ell$ such that the value of $M$ at $x_3$ is not $m_i$, otherwise, the first time the sequence $M_{i,{\ell-repeated}}$ appears in $M$ is at $x_2$, where $x_2 < x_1$, but we assumed that the first time $M_{i,{\ell-repeated}}$ appears in $M$ is at $x_1$ so we reach a contradiction. So there must be an index $x_3$ such that the value of $M$ in this index is different from the value of $M'$ at this index and so it must be that $M$ is different from $M'$.
\end{proof}

\noindent\relem{Lemma23}[repeated]
\lemlegalsys
\erelem

\begin{proof}
First, we know that there is a run that starts from the initial configuration and reached $S$. 
Assuming that when we reach $S$ the input sequence, $M$, of $p_{sender}$ is an infinite incremental sequence of messages. 
Let $R_S$ be the run that starts at $C_S$, feasible at $N_S$ and causes each party (except one party, $p' \neq p_{sender}$, that takes no steps in $R$, crashed) to deliver $M$. 
We know that $R_S$ exists because $S$ is legal System State so there must be a run feasible at $S$ satisfying the $SuRB$ conditions even when one party may crash and the input sequence of $p_{sender}$ is $M$. 
Let $C_{R_S}$ be a sequence of the System States obtained by applying $R_S$ on $N_S$. Let $CR$ be a sequence of System States from $C_{R_S}$ such that the next step that follows each of those System States is that $p_{sender}$ reads a value from the input queue. $CR$ contains an infinite number of system states. 
We know that the memory of each party is bounded and so is the links, so there must be at least one system state that repeats itself infinitely many times during $CR$, let's call it $S_{\inf} = (C_{\inf},N_{\inf})$.

\begin{figure}[t]
\center
\includegraphics[scale=0.6]{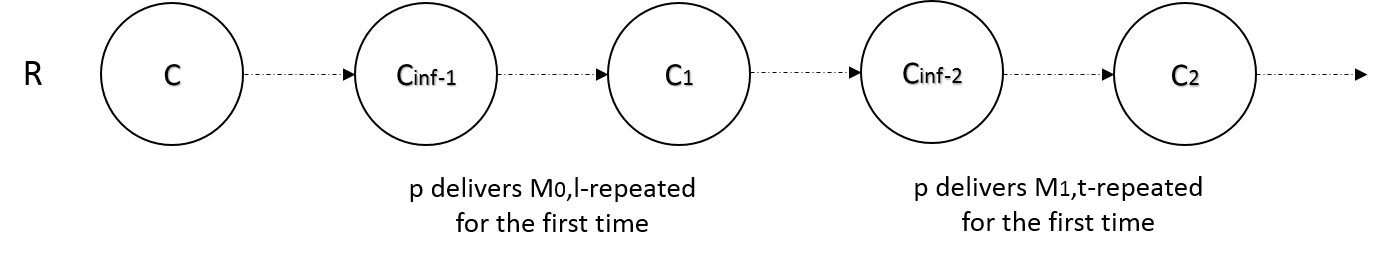}
\caption{In $R$ each honest party delivers an infinite incremental sequence of messages}
\label{fig:figure4}
\end{figure}

\begin{figure}[t]
\center
\includegraphics[scale=0.6]{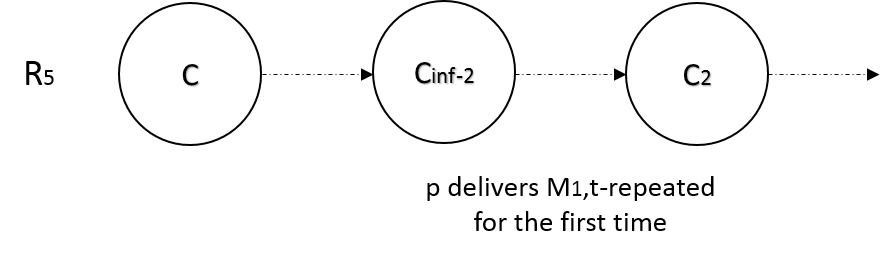}
\caption{Creating $R_5$ by concatenating $R_1$, $R_3$ and $R_4$ }
\label{fig:figure6}
\end{figure}

Let $S_{\inf-1} = (C_{\inf-1},N_{\inf-1})$ be the first time $S_{\inf}$ appears in $R_S$. Let $x_1$ be an index of a message, $m_0$, in $M$ such that in $S_{\inf-1}$, $p_{sender}$ didn't already read $m_0$. From \lemmaref{Lemma21} we know that there exists $\ell > 0$ and a sequence $M_{0,{\ell-repeated}}$ in $M$ such that the first time that $m_{0,{\ell-repeated}}$ appears in $M$ is after $x_1$. Let $S_1$ be a system state in $R_S$ in which $p_{sender}$ already read all the messages in $M_{0,{\ell-repeated}}$ and let $S_{\inf-2} = (C_{\inf-2},N_{\inf-2})$ be the first time $S_{\inf}$ appears in $R_S$ after $S_1$. Let $x_2$ be an index of a message, $m_1$, in $M$ such that in $S_{\inf-2}$, $p_{sender}$ already read $m_1$. From \lemmaref{Lemma21} we know that there exists $t > 0$ and a sequence $M_{1,{t-repeated}}$ in $M$ such that the first time that $M_{1,{t-repeated}}$ appears in $M$ is after $x_2$. Let $S_2 = (C_2,N_2)$ be a system state in $R_S$ in which $p_{sender}$ already read all the messages in $M_{1,{t-repeated}}$ (see  \figureref{fig:figure4}).

Let's now separate $R$ to four runs. The first, $R_1$, starts at  $C_S$, feasible at $N_S$ and ends in $C_{\inf-1}$. The second, $R_2$, starts at $C_{\inf-1}$, feasible at $N_2$ and ends in $C_{\inf-2}$, the third, $R_3$, starts at $C_{\inf-2}$, feasible at $N_2$ and ends in $C_{2}$ and the last, $R_4$, starts at $C_{2}$, feasible at $N_3$ and continues infinitely as $R_S$.

\begin{figure}[t]
\center
\includegraphics[scale=0.6]{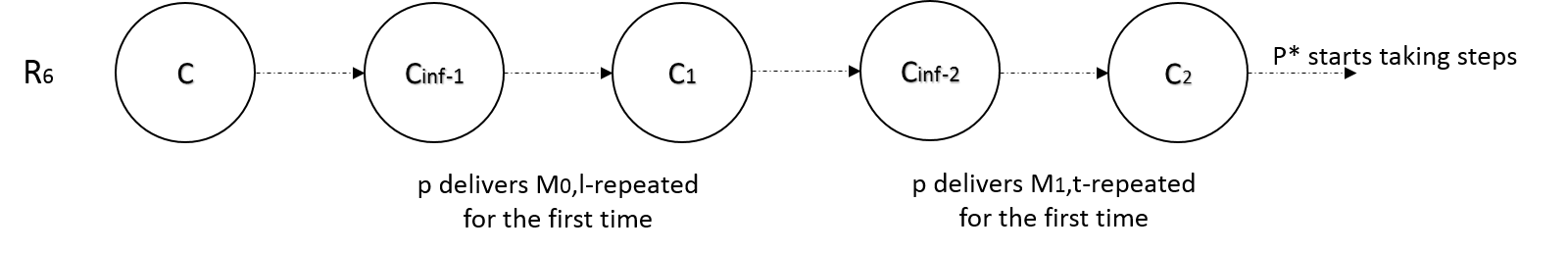}
\includegraphics[scale=0.6]{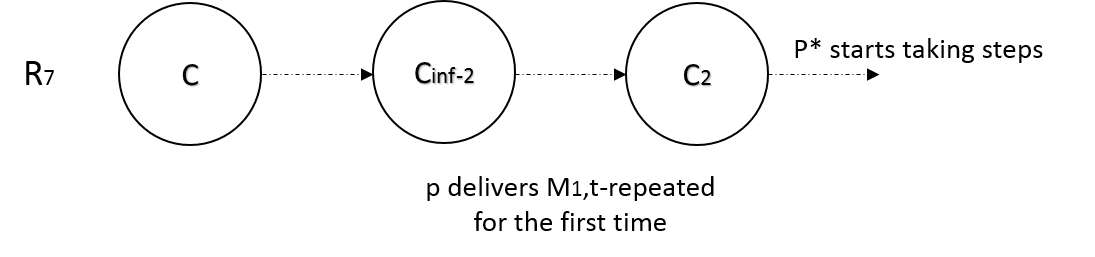}
\caption{Cut $R$ and $R_5$ in $C_2$ and receiving $R_6$ and $R_7$}
\label{fig:figure7}
\end{figure}

Now we can concatenate $R_1$, $R_3$, $R_4$ and receive $R_5$ feasible at $N_S$ such that in $N_SR_5$, $p$ does not deliver the sequence of messages that it delivered during $R_2$ which is not empty (see \figureref{fig:figure6}).

Now let's cut $R$ and $R_5$ in configuration $C_2$ and receive two new partial runs such that $p'$ is honest and does not take steps until we reach configuration $C_2$. We call those partial runs $R_6$ and $R_7$, respectively (see \figureref{fig:figure7}).

Let $M_{5}$, $M_{6}$ be the sequence of messages that $p$ delivers in $N_SR_5$, $N_SR$, respectively. Let $M_{R_2}$ be the sequence of messages that $p$ delivers during $R_2$. Now if $p'$ wakes up in $S_{2}$  it cannot distinguish between $R_5$ and $R$. There is a scenario in which $p'$ delivers $M_{5}$ while all other parties deliver $M_{6}$ and there is a scenario in which $p'$ delivers $M_{6}$ while all other parties deliver $M_{5}$. So it is enough to show that $M_{5}$ and $M_{6}$ are different. From \lemmaref{Lemma22} we know that those to sequences must be different so we are done.
\end{proof}

\commentout{
\rethm{thm21}[repeated]
\thmsurbtxt
\erethm
\begin{proof}
\ms{this is the old proof, it also appears in the article so can we remove it?}
Assume that there is a protocol that satisfies the SuRB properties even when one party is prone to crash. Let $PR$ be such a protocol. 
Let $S_{init} = (C_{init},N_{init})$ be the initial System State of $PR$. 
By lemma \ref{Lemma23}, there is a run $R$ starts from $C_{init}$ and is feasible at $N_{init}$ such that when applying $R$ on $N_{init}$ there is an honest party that does not deliver the same sequence as all other parties. 
If the changes in the delivered sequence is infinite we are done, otherwise lets look at the System State, $S_1 = (C_1,N_1)$, in which the changes in the delivered massages stops.
We know that $S_1$ is a legal System State, because we have a partial run that starts at the initial configuration, feasible at the initial Network State and reaches $S_1$. 
New we can apply again lemma \ref{Lemma23} on $S_1$ and obtain another partial run, $R_1$, that starts at the $C_1$, feasible at $S_1$ and contains two honest parties that do not deliver the same sequence of messages. 
We can do this process again and again. 
If finally we reach a point in which the changes in the delivered sequence is infinite we are done, otherwise we will construct an infinite run in which the changes in the delivered sequence is infinite.
\end{proof}

\section{Meaning and Consequences of the Impossibility Results}
\subsection{Do our Impossibility Results Contradict the Result Achieved in~\cite{SelfStabilizePaxos}}
\subsubsection{Practical Self-Stabilize Paxos overview~\cite{SelfStabilizePaxos}}
The main problem of stabilizing the original Paxos (\sectionref{PaxosConsensusandReliableBroadcast}) is the proposal number, which in the original Paxos is assumed to be an unbounded integer. If the system starts with some known configuration then it is reasonable to assume that, but in the case of self-stabilization the system could start in a configuration in which the proposal number variable are set to $max\_int$, and so we will not be able to increase it any more.

A finite labelling scheme intends to solve this problem. The main idea in such a system is that given $n$ elements, we can always find an element that is greater than each of those $n$ elements. This way we will always be able to increase the proposal number. The idea is to make sure the system will always reach a point in which the rest of the run is identical to the original Paxos and so will achieve all the properties of the original Paxos. The proposal number in self-stabilize Paxos is a tuple of size $n$ (the number of parties) such that each index in the tuple is also a tuple of size 5. The first element in the inner tuple is a finite labelling scheme, the second and the third are the step and the trial (such as in Paxos), the fourth is the proposer id, and the last one is called the cancelling label. This is the label that is greater than the label in the first index of the tuple (or null if no such label exists). If a proposer wants to increase a label, it first tries to increase the step or the trial, but if those are equal to max\_int it increases the label (remember that it is always possible to increase the label). The proof shows that eventually the system will reach a state in which the proposal number of each party has the same label and the step and trial will be small compared to the max\_int value. From this point, the protocol continues as the original Paxos in \sectionref{PaxosConsensusandReliableBroadcast}.

\subsubsection{Relation Between Practical Self-Stabilize Paxos and our Impossibility Results}
In~\cite{SelfStabilizePaxos} Peva Blanchard, Shlomi Dolev, Joffroy Beauquier and Sylvie Dela??t describe how they can build a self-stabilizing Paxos and receive a self-stabilizing distributed state machine in the bounded model. Paxos promise  that when there is a single honest leader all honest parties will agree on the same value. Replicated state machine using multiple instances of Paxos such that in the $i$-th Paxos instance all honest parties decide on the $i$-th action to perform and so all honest parties perform the exact same action in the exact same order.

It looks like this result contradicts the impossibility result of the SuRB problem. In the SuRB problem we have a single leader that wants to send a sequence of values to all honest parties. We showed that it is impossible to achieve a protocol that satisfies the SuRB conditions so how in~\cite{SelfStabilizePaxos} Peva Blanchard, Shlomi Dolev, Joffroy Beauquier and Sylvie Dela??t show a protocol that achieves a replicated state machine in which the honest leader sends a sequence of actions (values) to all honest parties.

Although it looks like a contradiction it actually is not. Our result actually reinforces the protocol presented in~\cite{SelfStabilizePaxos}. In~\cite{SelfStabilizePaxos} they didn't consider the self-stabilization as the original definition of stabilization instead they define what they called practically self-stabilization. In practically self-stabilization we do not require that the stabilize run lasts forever but only long enough for any concrete system'??s timescale. If an honest party does not participate all along the stabilize run then it is not guaranteed that he will perform the exact same action as all other honest parties. The practically self-stabilization approach was first introduced in~\cite{WhenConsensusMeetsSelfStabilization} where a practically self-stabilizing consensus protocol achieved in an asynchronous shared memory environment.

In the self-stabilizing Paxos the run is divided into epochs. Each epoch is considered as a stabilized run. When an epoch terminates, the history of the replicated state machines is  cleared and a blank history is started. The meaning of the stabilization in this article is that the system always reaches a very long epoch. This epoch is considered infinite, it is like counting from one to $2^{64}$. It is mentioned that counting from one to $2^{64}$ and increasing the counter each nanosecond will take us about 500 years to finish. So if the system always converges to such an epoch, practically we do not have to worry that some honest party will miss this long epoch and sleep for 500 years. But if some honest party will  sleep for 500 years he will miss the entire epoch and the history of the actions performed in this epoch will disappear from the system.
This means that the sleepy party will not be able to overcome the actions that have been lost, just like we proved in the SuRB impossibility result. In addition the convergence time of the protocol might be as long as the stabilize run. If the run starts from some configuration in which the epoch is in its youth, it will take about 500 years for the epoch to terminate which means that in this time our history will contain fake values, those values are located in the history because of transient faults. It will take about 500 years for them to disappear from the system.
\dd{should we discuss the convergence time of that paper?}
\ms{fixed}


\subsection{The  SuRB Impossibility Insight}
The  SuRB impossibility is very interesting. It shows that self-stabilization in not a trivial property that can be achieved under any conditions. It also shows that there is no technique that takes a distributed protocol and stabilizes it. Notice that without the transient fault and the need to self-stabilize we could solve this problem with the original Reliable Broadcast. Also when Byzantine does not come into consideration we could also solve this problem by using the SuRB protocol we showed earlier in this article. 
Additionally it is clear that without  asynchrony we wouldn't be able to prove the impossibility result. 
So all these three properties are necessary to prove the impossibility result and without any one of them the problem is  solvable.

One more important insight must be clearly specified. Unlike the SuRB problem the  SuRB problem does not require  each honest party to deliver an infinite sequence of messages. Theoretically if someone could find a protocol that satisfies the  SuRB conditions and stop delivering messages in some point, then it is absolutely OK. The  SuRB impossibility result shows that no such protocol exists. But it is important to understand that the requirement that each honest party will deliver an infinite sequence of messages does not come from the problem definition as in the SuRB problem but from the system model. This is why the impossibility result here is more surprising than the impossibility result of the SuRB problem. In the SuRB problem we can intuitively understand that we cannot save the delivered sequence forever and at some point we must forget what messages were delivered, while in the  SuRB this is not the case, we only have to agree on one message and stop. Our proof shows that we cannot agree on a single message when we consider asynchrony, self-stabilization, and Byzantine assumptions.

\newpage
\section{Further Research}
There are further research questions that can be asked regarding the following result, for example:

\beginsmall{enumerate}

\item What is the minimal assumption that is required in order to make those problems solvable?
\item Is it possible to build a self-stabilizing mechanism that replaces digital signature (such as Lamport uses the original Reliable Broadcast to replace digital signature~\cite{ByzantinePaxos})?
\item Is it possible to Byzantine the practically self-stabilize Paxos presented in~\cite{SelfStabilizePaxos}? 
\end{enumerate}
}

\section{ Self-Stabilize Reliable Broadcast}\label{sec:ssurb}
\hide{
\noindent\rethm{thm:linklayer}[repeated]
\thmlinklayer
\erelem
\begin{proof}
\begin{claim}[Link Layer Claim 1 - Sending termination (condition 1)]
If $P_{1}$ is honest then the sending function eventually terminates.
\end{claim}

The send function is sending the given message $\bar c + 1$ times. So eventually it terminates.

\begin{claim}[Link Layer Claim 2 - Non-triviality (condition 3)]
If $P_{1}$ and $P_{2}$ are honest and $P_{1}$ sends $m$  infinitely many times in a row then $P_{2}$ will deliver $m$ infinitely many times.
\end{claim}

From the link properties we know that if $P_{1}$ sends $m$ infinitely many times in a row then $P_{2}$ will receive $m$ infinitely many times. Which means that $P_{2}$ will see $m$ $\bar c + 1$ times in a row infinitely many times and will deliver $m$ infinitely many times.

\begin{claim}[Link Layer Claim 3 - No duplication (condition 2.b)]
Let $P_{1}$ and $P_{2}$ be two honest parties. Assume $P_{1}$ sends $m$ to $P_{2}$ $k$ times. $P_{2}$ will deliver $m$ no more than $k$ times (unless $m$ was delivered by $P_{2}$ as a ghost message but this can happen no more than 4 times as we will see later).
\end{claim}

If $P_{1}$ sends $m$ to $P_{2}$ $k$ times  then the message is  actually being sent $m$ $(\bar c + 1)k$ times. $P_{2}$ delivers $m$ when it sees it $\bar c + 1$ times in a row. The communication layer does not duplicate messages so $P_{2}$ will receive $m$ no more than $(\bar c + 1)k$ times and will deliver $m$ no more than $k$ times.

\begin{claim}[Link Layer Claim 4 - No reordering (condition 2.d)]
Let $P_{1}$ and $P_{2}$ be two honest parties and let $m_{1}$ and $m_{2}$ be messages sent by $P_{1}$ to $P_{2}$. If $P_{1}$ sends $m_{1}$ before $m_{2}$ then $P_{2}$ will not deliver $m_{2}$ before $m_{1}$.
\end{claim}

Assume that $P_{1}$ sends $m_{2}$ after it sends $m_{1}$. Once $P_{1}$ starts sending $m_{2}$, it stops sending $m_{1}$ so when $P_{1}$ starts sending $m_{2}$ there is at most $\bar c$ $m_{1}$ messages that may arrive to $P_{2}$ after $m_{2}$. So when $m_{2}$ first arrives to $P_{2}$ there is not enough messages left in the communication link for $m_{1}$ to be delivered by $P_{2}$. So if $m_{1}$ was delivered by $P_{2}$ it must be before $m_{2}$ is delivered.

\begin{claim}[Link Layer Claim 5 - At most three ghost messages (condition 2.c)]
Let $P_{1}$ and $P_{2}$ be two honest parties and let $m_i (i = 1, ... ,  4)$ be four messages delivered after the last transient fault by $P_{2}$ (in this order), then $m_4$ must be a real message that was sent by $P_{1}$.
\end{claim}

\noindent\textbf{First Ghost Message:}
Because of the transient fault $P_{2}$ can start believing that it has to deliver $m_{1}$, so  $m_{1}$ may be fake.\\
\\
\textbf{Second Ghost Message:}
Because of transient fault $P_{2}$ can start the run when its $LastMessage$ are set to $m_{2}$ and its counter is set to 1 (or higher). Also there could be $\bar c$ $m_{2}$ messages in the communication link. So $m_{2}$ may also be fake.\\
\\
\textbf{Third Ghost Message:}
Because of the transient fault $P_{1}$ can start the run inside the loop of the send function believing it is sending $m_3$ and so $m_3$ may be received by $P_{2}$. So the third message may also be fake.\\
\\ 
\textbf{The Fourth message cannot be Fake:}
Assume $m_4$ is also fake. That means that $P_{1}$ did not send $m_4$. $m_4$ must be in the system somehow before it was received by $P_{2}$. It is not in the $LastMessage$ variable of $P_{2}$ (because it was caught by $m_3$). So $m_4$ must arrive from the communication link but in this case there are at most $\bar c$ $m_4$ messages that may arrive so $m_4$ cannot be delivered by $P_{2}$. And so $m_4$ cannot be a fake message.
\end{proof}

\noindent\rethm{thm:SuRB-crash}[repeated]
\thmsssurb
\erelem

\begin{proof}

We first prove the S2 property of SuRB:
If the sender is honest and the input stream contains a single message $m$ repeated forever, then each honest party must have a suffix of delivered messages that contains only $m$.
By the link layer properties, we know that at most three messages could be ghost messages. So from the fourth message on, all the messages delivered by each honest party must be real. $p_{sender}$ is honest so he sends $m$ infinitely many times. 
The link layer properties imply that eventually each honest party will receive $m$ and only $m$, which means that eventually each honest party will deliver $m$ and only $m$. In addition, after each honest party delivered three messages, at the latest, we can be sure that all following messages that will be delivered by every honest party will be $m$ and only $m$.

We now prove the S1 property.
If $p_{sender}$ is honest,  from S2 we  conclude that eventually each honest party will deliver only the message that has been given to the sender via the input stream and so all honest parties will share a common suffix. 

If $p_{sender}$ is has crashed (dishonest) then eventually each honest party will stop receiving messages from $p_{sender}$ and will stop delivering messages. This means that there is no party that delivers an infinite number of messages and so condition S1 clearly holds, since it considers only the case in which there is an honest party delivering an infinite number of messages.
\end{proof}
}

\subsection{ Self-Stabilize SuRB Impossibility with a Single Byzantine Party}\label{sec:ssurb-lower-apndx}


\def\lemnewmsg{
Assume that the input stream given to $p_{sender}$ is a stream of a single message repeated forever. Let $PR$ be a protocol that satisfies the conditions of the  SuRB problem and is resistant to one Byzantine party. Let $S = (C,N)$ be an arbitrary System State. For each party $p' \in P$,  $p'\not=p_{sender}$, there is a run $R$ that starts at $C$ and is feasible at $N$ such that $p'$ takes no steps and each honest party delivers a new message.
}

\begin{lemma}\label{Lemma6}
\lemnewmsg
\end{lemma}
\begin{proof}
Assume that $p'$ is Byzantine and takes no steps at all. $PR$ is tolerant to one Byzantine party so it must satisfy the  SuRB conditions even when $p'$ takes no steps ($p'$ might be Byzantine). 
Assume that there exists a System State $S = (C,N)$ such that for each run $R$ that starts at $C$ and is feasible at $N$ there is a party $p$ that does not deliver a new message. 
This contradicts condition 2 of  SuRB, because each honest party must have a suffix of delivered messages that contains only the input of $p_{sender}$, the run may start at $S$ and so $p$ will not deliver any messages at all in particular $p$ will not deliver the message sent by $p_{sender}$. So from any System State $S$ there is always a run $R$ such that $p'$ takes no steps in $R$ and each honest party delivers a new message.
\end{proof}

\noindent\relem{Lemma7}[repeated]
\leminf
\erelem
\begin{proof}
This is a direct result of \lemmaref{Lemma6}. Let $S = (C,N)$ be an arbitrary System State. We know that there is a run $R_1$ that starts at $C$ and is feasible at $N$ which causes each honest party to deliver a new message and $p'$ takes no steps in $R_1$. Let $C_{R_1}$ be the configuration in $R_1$ in which each honest party already delivered a new message and let $N_{R_1}$ be the corresponding network state. 
Again from \lemmaref{Lemma6} we know that there is a run $R_2$ that starts in $C_{R_1}$ and is feasible at $N_{R_1}$ which causes each honest party to deliver a new message and $p'$ takes no steps in $R_2$. 
Let $C_{R_2}$ be the configuration in $R_2$ in which each honest party already delivered a new message and let $N_{R_2}$ be the corresponding network state. 
We can now repeat the same arguments over and over again and receive $R_3, ...,  R_n, ...$. 
So we have an infinite sequence of runs such that each honest party delivers at least one new message in each run and each run starts with the same configuration in which the last run ended. 
In addition for each $i$ the network state reached by applying $R_i$ on $N_i$ is $N_{i+1}$. 
We can now concatenate all these runs and receive a new run that starts at $C$ and is feasible at $N$ in which each honest party delivers an infinite number of messages and $p'$ takes no steps in it. 
\end{proof}
\noindent\relem{Lemma8}[repeated]
\lemconcat
\erelem

\begin{proof}
If $p_{sender}$ is Byzantine it can change its internal state whenever it wants, alternatively, by using a transient fault we can also change the internal state of $p_{sender}$ whenever we want. Let's look at $R_1$ and assume that when we reach configuration $c^{(1)}_n$, $p_{sender}$ changes its internal state to be the same as in $c^{(2)}_1$ (could be because $p_{sender}$ is Byzantine or because a transient fault happened).
Let's call this new configuration $c^{(1)}_{fake}$ and let's call this new run $R_{fake}$. 
Notice that we assumed that $c^{(1)}_n$ and $c^{(2)}_1$ are different only in the internal state of $p_{sender}$, so we can conclude that $c^{(1)}_{fake} = c^{(2)}_1$. 
We assumed that $R_2$ is feasible at $N_2$. In $R_{fake}$ we only changed the internal state of $p_{sender}$, this step does not change the network state, so the resulting network state of applying $R_{fake}$ on $N_1$ is also $N_2$. Now we can concatenate $R_{fake}$ and $R_2$ and receive a new run $R = [c^{(1)}_1,s^{(1)}_1, ...,  c^{(1)}_n,c^{(2)}_1,s^{(2)}_1, ...,  c^{(2)}_n]$ feasible at $N_1$.
\end{proof}

\begin{figure}[ht]
\begin{minipage}[b]{0.45\linewidth}
\centering
\includegraphics[scale=0.6]{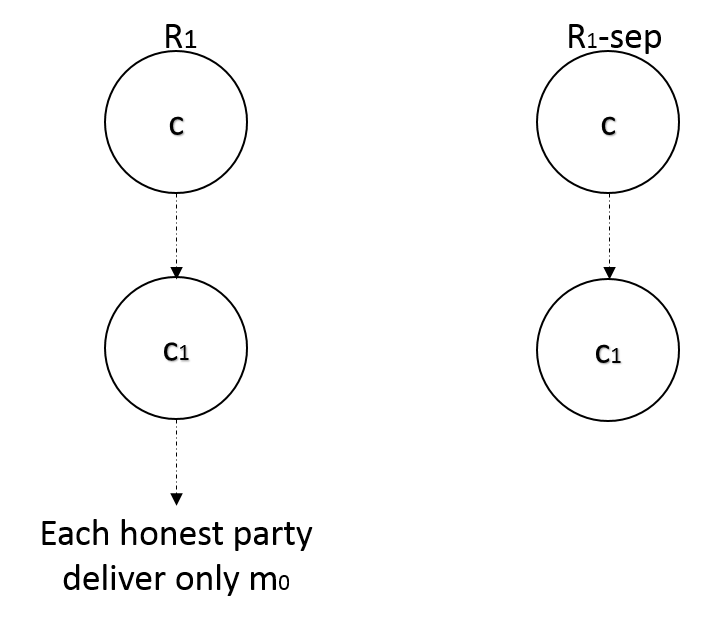}
\caption{In $C_1$ each honest party delivers only $m_1$}
\label{fig:figure1}
\end{minipage}
\hspace{0.5cm}
\begin{minipage}[b]{0.45\linewidth}
\centering
\includegraphics[scale=0.6]{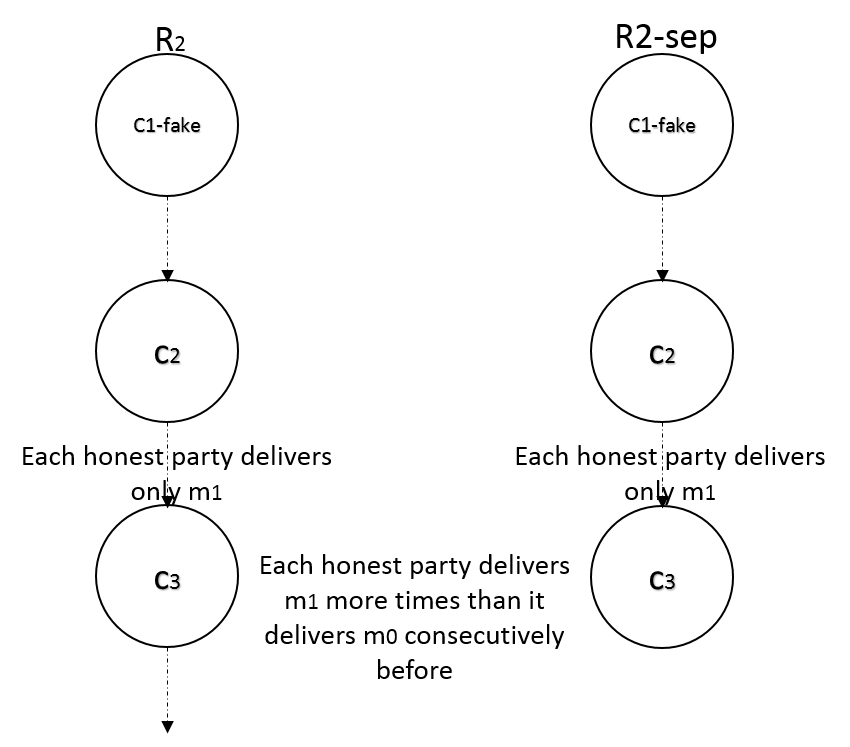}
\caption{In $C_3$ each honest party delivers $m_2$ more times than it delivers $m_1$ consecutively before}
\label{fig:figure2}
\end{minipage}
\end{figure}

\begin{figure}[t]
\centering
\includegraphics[scale=0.6]{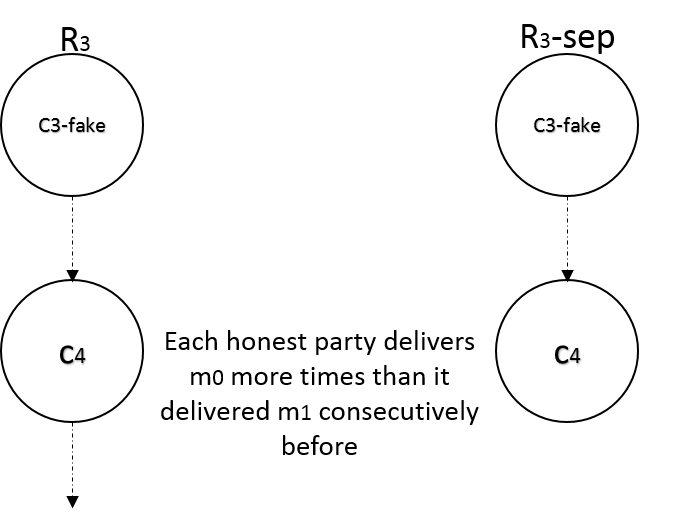}
\caption{In $C_4$ each honest party delivers $m_1$ more times than it delivers $m_2$ consecutively before}
\label{fig:figure3}
\end{figure}

\noindent\relem{Lemma9}[repeated]
\lemssinfinc
\erelem

\begin{proof}
Let $S = (C,N)$ be an arbitrary System State and $p'$ be a party that is not $p_{sender}$. By \lemmaref{Lemma7},  there is a run, $R_1$, that starts from $C$ and is feasible at $N$ such that $p'$ takes no steps in $R$ and each honest party delivers an infinite number of messages. 
Let's look at the sequence of configurations and network states obtained by applying $R_1$ on $N$. By condition S2 of SuRB, 
each honest party must have a suffix of the delivered messages that contains only the message given to $p_{sender}$ via the input stream. 
Assume that $p_{sender}$ is honest and the messages given to him via the input stream is $m_0$. 
Eventually we will reach a System State $S_1 = (C_1,N_1)$ such that in $S_1$, each honest party already performed deliver $m_0$ and from $S_1$ the only message that each honest party will deliver is $m_0$, also each honest party will deliver $m_0$ infinitely many times. 
Now we can stop $R_1$ in configuration $C_1$ and receive $R_{1-sep}$ (see \figureref{fig:figure1}).

Let $C_{1-fake}$ be the configuration that is different from $C_1$ only in the internal state of $p_{sender}$, Let $S_{1-fake} = (C_{1-fake},N_1)$. 
Assume that in $C_{1-fake}$ each time $p_{sender}$ reads from the input stream a transient fault causes the  message read to change inside $p_{sender}$ memory to $m_1$ (just like $p_{sender}$ actually reads $m_1$ instead of $m_0$). 
Let's start from $S_{1-fake}$. We know that there is a run, $R_2$, that starts from $C_{1-fake}$ and is feasible at $N_1$ such that $p'$ takes no steps in $R_2$ and each honest party delivers an infinite number of messages (\lemmaref{Lemma7}), also there is a suffix of delivered messages that contains only the input of $p_{sender}$ (condition S2). 
Eventually we will reach a System State $S_2 = (C_2,N_2)$ such that from $S_2$ the only message that each honest party will deliver is $m_1$ and each honest party will deliver $m_1$ infinitely many times. 
Because each honest party will deliver $m_1$ infinitely many times, eventually we will reach a System State in which each honest party delivered $m_1$ more times than it delivered $m_0$ consecutively before. 
Let's call this System State $S_3 = (C_3,N_3)$. Now we can stop $R_2$ in configuration $C_3$ and receive $R_{2-sep}$ (see \figureref{fig:figure2}).

We can repeat the process from $C_3$ and receive a new System State $S_{3-fake} = (C_{3-fake},N_3)$ and a run $R_3$ that starts in $C_{3-fake}$ and is feasible at $N_3$ such that each honest party will have a suffix of delivered messages that contain only $m_0$ (we assume that in $C_{3-fake}$, $p_{sender}$ reads the real value, $m_0$), also there must be a System State $S_4$ in which each honest party already delivered $m_0$ more times than it delivered $m_1$ consecutively before. Again we can obtain $R_{3-sep}$ in the same way (see \figureref{fig:figure3}).

We can repeat the process over and over again and receive an infinite sequence of runs such that when applying $R_{i-sep}$ on $N_i$, each honest party delivered $m_j$ consecutively more times than it delivered $m_{1-j}$ consecutively before ($j \in \{1,2\}$). Also for each $i > 0$ the last configuration in $R_{i-sep}$ and the first configuration in $R_{{i+1}-sep}$ are different only in the internal state of $p_{sender}$ and their network state are the same so by \lemmaref{Lemma8} we can concatenate all those runs (by using transient faults or maliciousness of the sender) and receive a run, $R_{final}$, feasible at $N$ such that $p'$ takes no steps in $R_{final}$ and when applying $R_{final}$ on $N$ each honest party delivers an infinite incremental sequence.
\end{proof}

\end{document}